\documentclass[final,twocolumn]{IEEEtran}

\usepackage{amsmath,epsfig,amssymb,algorithm,amsthm,cite,url}

\usepackage{algpseudocode}
\usepackage{hyperref}
\newtheorem{theorem}{Theorem}
\newtheorem{lemma}[theorem]{Lemma}
\newtheorem{assumption}{Assumption~A-\kern-0pt}
\newtheorem{proposition}[theorem]{Proposition}
\newtheorem{corollary}[theorem]{Corollary}

\newtheorem*{remark*}{Remark}
\DeclareMathOperator{\tr}{tr}

\newcommand{\bQ}{{\bf Q}}
\newcommand{\bT}{{\bf T}}

\usepackage{pgfplots}
\newcommand{\bOmega}{\boldsymbol{\Omega}}
\newcommand{\bTheta}{\boldsymbol{\Theta}}
\newcommand{\bSigma}{\boldsymbol{\Sigma}}
\newcommand{\bxi}{\boldsymbol{\xi}}

\usepackage{amssymb}

\let\emptyset\varnothing
\title{On the Smallest Eigenvalue of General correlated Gaussian Matrices}
\author{{Abla Kammoun, ~\IEEEmembership{Member, ~IEEE } and
Mohamed-Slim Alouini, ~\IEEEmembership{Fellow, ~IEEE,}
}
}
\begin{document}
\maketitle
\begin{abstract}

%This paper analyzes the behaviour of the smallest eigenvalue  of generally correlated Gaussian matrices. This questi 
%The analysis of the spectrum of covariance matrices is a question which arises in various applications of signal processing. 
%Controlling the spectrum of 
 This paper investigates the  behaviour of the spectrum of generally correlated Gaussian random matrices whose columns are zero-mean independent vectors but have different correlations, under the specific regime where the number of their columns and that of their rows grow at infinity with the same pace.  This work  is, in particular, motivated by applications from statistical signal processing and wireless communications, where this kind of matrices naturally arise. 
 Following the approach proposed  in \cite{LOU10}, we prove that under some specific conditions, the smallest singular value of  generally correlated Gaussian matrices is almost surely away from zero. 
%Very recently, a method based on Gaussian calculus has been proposed in order to control the behaviour of non-centered Gaussian matrices. 
\end{abstract}

\section{Introduction}
Let $\bSigma_n$ be a rectangular random matrix of size $N\times n$. The study of the behaviour of the asymptotic spectrum of $\bSigma_n$ when $N,n\to+\infty$ has been investigated in several works. As is known, when the elements of $\bSigma_n$ are zero-mean and unit variance independent and identically distributed (i.i.d.) and $\frac{N}{n}\to c<1$,  the empirical measure of the eigenvalues of $\frac{1}{n}\bSigma_n\bSigma_n^*$ converge weakly to a deterministic probability distribution which is supported  by the interval $\left[(1-\sqrt{c})^2,(1+\sqrt{c})^2\right]$ \cite{marchenko}. A question which immediately arises in connection with this result concerns the asymptotic behaviour of the extreme singular values. 
At first sight, one would expect  the smallest and the largest eigenvalues of $\frac{1}{n}\bSigma_n\bSigma_n^*$ to converge to $(1-\sqrt{c})^2$ and $(1+\sqrt{c})^2$, respectively. While this statement is correct, it cannot be directly inferred from the aforementioned weak convergence result. As a matter of fact, the proof generally requires the use of more advanced techniques improving the weak convergence result. %which further refine the weak convergence result. 
%Obviously, this cannot be inferred from the limiting spectral distribution.
 First findings related to these issues can be traced back to the works of J. Silverstein \cite{SIL85} and S. Geman \cite{GEM80}, who provided a rigorous proof showing that the extreme eigenvalues  of $\frac{1}{n}\bSigma_n\bSigma_n^*$ converge in the Gaussian case to the edges of the limiting support $(1-\sqrt{c})^2$ and $(1+\sqrt{c})^2$. This result was then extended to the case of non-Gaussian matrices but with independent and identically distributed entries \cite{BAI93}. The characterization of the limiting support of $\bSigma_n$ is much more difficult in the case where the column entries of $\bSigma_n$ are correlated. Instead of determining the exact support, many works focused on establishing the almost sure absence of eigenvalues of $\frac{1}{n}\bSigma_n\bSigma_n^*$ in any closed interval outside the support of the limiting distribution. 
We can cite, for sake of illustration, the work of \cite{BAI99} applying for the simple-correlated case where the columns of $\bSigma_n$ are correlated with the same correlation matrix and that of \cite{LOU10} which deals with non-centered uncorrelated models. 

In many applications, this result, though limited, is essential. It can be, for instance,  used to efficiently handle random quantities involving the Gram matrix $\frac{1}{n}\bSigma_n\bSigma_n^*$ or its inverse.  

%The case where the columns of $\bSigma_n$ are independent vectors with the same correlation matrix was finally investigated in \cite{BAI99}, where it has been proven that almost surely, no eigenvalue of $\frac{1}{n}\Sigma_n\Sigma_n^*$ is outside the limiting support. 

% It is worth noting that in parallel to these works, different approaches that provides bounds for the extreme singular values holding with high probability were proposed. Such approaches are merely based on the observation that the control of the least singular value is equivalent to control the distance of one column of matrix $\bSigma_n$ and the hyperplane spanned by the other columns \cite{Litvak}. As it has been noted in \cite{BookTAO}, these methods strongly rely on the independence on all the entries of $\bSigma_n$, and cannot be easily extended to more general settings. 
In this paper, we consider the generally correlated Gaussian model in which the columns of $\bSigma_n$ are zero-mean independent Gaussian random vectors but with different correlations. First results related to this model are due to Wagner et al. \cite{WAG10} who characterize the asymptotic behaviour of the limiting distribution of $\frac{1}{n}\bSigma_n\bSigma_n^*$. This result was in particular applied to the analysis of the performance of the regularized-zero forcing linear precoding technique \cite{WAG10}.

 Since then, this model has known an increasing popularity, mostly spurred by applications in multi-user nultiple-input-single-output (MISO) systems \cite{caire-13,Kammoun2014b} and the very recent robust signal processing applications \cite{alouini-14}. In what follows,  we provide two different applications where the general correlation Gaussian model arises.

% We propose to establish that the eigenvalues of the Gram matrix of $\frac{1}{n}\bSigma_n\bSigma_n^*$ are almost surely in the limiting 
%Such model is motivated by recent applications in Massive MIMO multi-user systems and also the most recent robust signal processing applications. We  
%From a practical point of view, the assumption of independent and identically distributed entries is highly restrictive. Random matrices that are encountered in statistical signal processing or in wireless communications are highly correlated. We provide hereafter two applications where  Gaussian matrices with columns having different correlations naturally arise. 
\begin{paragraph}{Multiple Input Single Output Channel}
	Consider the downlink of a single-cell system in which a base station (BS) with $N$ antennas serves $n$ users equipped each with a single antenna each and assume that $N<n$. The downlink channel vector ${\bf h}_k$ between the BS and the $k$ th user is given by \cite{WAG10}:
	$$
	{\bf h}_k ={\bf R}_k^{\frac{1}{2}}{\bf z}_k.
	$$
	with ${\bf z}_k$ is a standard complex Gaussian vector and 
	matrix ${\bf R}_k$ is essentially function of the richness of the scattering between the BS and the user of interest and as such is specific for each user. To mitigate inter-user interference, the BS precodes the transmitted signal by a matrix ${\bf G}$ which depends on the channel conditions for all users. Among the used precoding techniques, we can cite the Zero-forcing (ZF) precoding given by \cite{HOC02}:
	$$
	{\bf G}=\left(\frac{1}{n}{\bf H}{\bf H}^*\right)^{-1}{\bf H},
	$$
	where ${\bf H}=\left[{\bf h}_1,\cdots,{\bf h}_n\right]$.
	The ZF precoding involves the inversion of the Gram matrix ${\bf H}{\bf H}^*$, a  step which becomes  critical in case the smallest eigenvalue is near zero. 
	In order to analyze the performance of using the ZF precoding, the regime under which the number of antennas $N$ and the number of users $n$ increase with the same pace is often assumed. %It is knowing an increasing popularity which has been spurred by the recent advances of random matrix theory tools. 
	The performance of the ZF precoding under this regime has been studied in \cite{WAG10}, where it has been assumed that the smallest eigenvalue of $\frac{1}{n}{\bf H}{\bf H}^*$ is bounded away from zero for all large $N$ and $n$. Although this assumption holds true for specific cases where all matrices ${\bf R}_k$ are equal, there is no proof supporting its validity in general. This is the reason why the authors in \cite{WAG10} opted to add it as an assumption, which is likely to always hold true and thus is unnecessary.
	%which can be attributed to the recent advances of tools borrowed from random matrix theory.   
%	The recent works that study the performance of the ZF precoding under this regime, used to assume that the smallest eigenvalue of ${\bf H}{\bf H}^{\mbox{\tiny H}}$ is bounded away from zero for all large $N$ and $n$. 
	
	%one has to ensure that  the smallest eigenvalue of ${\bf H}{\bf H}^{\mbox{\tiny H}}$ is bounded away from zero for all large $N$ and $n$. 
\end{paragraph}
\begin{paragraph}{Robust Statistics}
	Consider a temporal series of $n$ vector observations ${\bf y}_1,\cdots,{\bf y}_n$ of size $N\times 1$. Assume that the contribution of each ${\bf y}_i$ can be decomposed as the sum of a useful signal plus an elliptical noise, i.e, 
	\begin{equation}
	{\bf y}_i= {\bf s}_i +{\bf x}_i,
	\end{equation}
	where ${\bf s}_1,\cdots,{\bf s}_n$ are Gaussian independent $N\times 1$ random Gaussian vectors with covariance ${\bf R}$ and ${\bf x}_i$ is drawn from a Compound Gaussian distribution, i.e,
	\begin{equation}
	{\bf x}_i=\sqrt{\tau}_i {\bf z}_i,
	\end{equation}
	where ${\bf z}_i$ are standard complex Gaussian vectors and $\tau_1,\cdots,\tau_n$ are scalar positive-valued random variables. We consider the problem of estimating the covarince matrix of ${\bf x}_i$. 
	In order to mitigate the impact of the heavy-tailed distributed noise, the use of robust covariance estimates known also as robust scatter estimates has been proven to be a good solution. These are given as the unique solution of the following equation:
	\begin{equation}
	\hat{\bf C}_N=\sum_{i=1}^n u({\bf x}_i^* \hat{\bf C}_N^{-1}{\bf x}_i){\bf x}_i{\bf x}_i^*,
	\end{equation}
	where $x:\mapsto u(x)$ is a scalar functional satisfying certain conditions \cite{couillet-pascal-2013}. In a recent submitted work, we prove that matrix $\hat{\bf C}_N$  converges in the operator norm to $\hat{\bf S}_N$ where $\hat{\bf S}_N$ is given by:
	\begin{equation}
	\hat{\bf S}_N=\sum_{i=1}^n v(\delta_i) {\bf x}_i{\bf x}_i^{*},
	\end{equation}
	with ${\delta}_1,\cdots,\delta_n$ are solutions of some fixed point equations \cite{alouini-14}. %independent of ${\bf z}_i$ and are solutions to the following system of equations:
	%$$
	%\delta_i=\frac{1}{N}\tr \left({\bf R}+\tau_i {\bf I}_N\right)
	%$$
	Conditioning on $\tau_i$, matrix ${\bf S}_N$ follows the model of generally correlated Gaussian matrices. The proof in \cite{alouini-14} relies on the control of the smallest eigenvalue of $\hat{\bf S}_N$. 
%These two examples show that the involved model considered in this paper is far from being infrequent, especially with the emergence of new signal processing applications involving elaborated models. We deeply believe that our result can help interested researchers handle in a more efficient way general correlated models. Our tools are also 

%This requirement is among the  %a requirement which 
%In addition to these two examples, we believe that this work can be useful for any other works dealing with 
	\end{paragraph}
Despite its importance, the generally correlated Gaussian model has not been extensively explored, most probably because of its recent emergence as a major practical model. Several questions related to the behaviour of the eigenvalues remain unanswered. 
A major question, illustrated by the two examples above, and which triggered our motivation for this work, concerns the control of the smallest eigenvalue of the Gram matrix $\frac{1}{n}\bSigma_n\bSigma_n^*$. Knowing that the smallest eigenvalue stay away of zero in the i.i.d case when $N<n$, one  can expect the same behaviour to hold for the general Gaussian correlated case under probably some mild conditions on the correlation matrices. In this paper, we provide a rigorous proof for this statement by essentially building on the techniques developed by  \cite{LOU10}. 

 %Up to the authors' knowledge, this question has been handled so far by adding a non-necessary assumption.

%Very recently, a new approach based on Gaussian calculus inspired from  the works of \cite{Capitaine-10} considers the almost sure 

%A natural question that arises from this result is to ask for the behaviour of extreme singular values.   

%{\bf Applications}: 
%{\bf Notations}:
%The set $\mathcal{C}_c^{\infty}$ denote the set of compactly supported real-valued functions defined on $\mathbb{R}$.  
%\begin{paragraph}
%{\it MIMO channel}
%
%\end{paragraph}
%\begin{paragraph}
%{\it Robust statistic methods}
%\end{paragraph}
%\begin{paragraph}{Notations}
%\end{paragraph}
\section{Problem statement and review of some results}
%\subsection{Problem statement}
All along the paper, we consider integers $n,N,\overline{N}$ such that $n \geq {N}$ and $\overline{N}\geq N$. We denote by $c_N$ the ratio $\frac{N}{n}$. We make the following assumptions:
\begin{assumption}
\label{ass:regime}
\begin{equation}
0<\liminf c_N \leq \limsup c_N <1. 
\end{equation}
\end{assumption}
The objective of this paper is to provide  some interesting properties  of the spectrum of generally correlated Gaussian matrices, i.e matrices whose columns are zero-mean independent random vectors but have different covariances. 
 Throughout this paper,
 matrix $\bSigma_n$ represents the complex-valued $N\times n$ matrix given by:
\begin{equation}
\bSigma_n=\left[ \bxi_1,\cdots, \bxi_n\right],
\end{equation}
where $\bxi_1,\cdots,\bxi_n$ are assumed to satisfy the following assumptions:
\begin{assumption}
\label{ass:boundedness}
$\left(\bxi_i\right)_{i=1}^n$ are zero-mean complex  Gaussian vectors of size $\overline{N}\times 1$ with covariance $\bTheta_i$ where $\left(\bTheta_i\right)_{i=1}^n$ is a sequence of  $N\times \overline{N}$ matrices verifying:
\begin{align}
	&w_{\rm min}\triangleq\inf_N \min_{1\leq i\leq n}\lambda_1\left(\bOmega_i\right) > 0,\\
	&w_{\rm max}\triangleq\sup_N \max_{1\leq i\leq n}\lambda_N\left(\bOmega_i\right) <+\infty,
\end{align}
where $\bOmega_i\triangleq \bTheta_i\bTheta_i^*$ and $\lambda_1(\bOmega_i)$ and $\lambda_N(\bOmega_i)$ are the smallest and largest eigenvalues of $\bOmega_i$.
\end{assumption}
We denote in what follows by $\lambda_1\leq \cdots\leq \lambda_N$ the eigenvalues of $\frac{1}{n}\bSigma_n\bSigma_n^*$. The empirical eigenvalue distribution of $\frac{1}{n}\bSigma_N\bSigma_N^*$ is defined as:
\begin{equation}
\hat{\mu}_N =\frac{1}{N}\sum_{k=1}^N \delta_{\lambda_k}.
\label{eq:measure}
\end{equation}
In order to characterize the asymptotic behaviour of $\hat{\mu}_N$, it is in practice quite common to analyze that of its Stieltjes transform (ST). 
Since the ST of a positive finite measure $\mu$ is given by:%If $\mu$ is a positive finite measure (i.e, $\mu(\mathbb{R})<+\infty$), the Stieltjes transform of $\mu$ is the function of complex variable given by:
$$
\Psi_{\mu}(z)=\int_{\mathbb{R}} \frac{d\mu(\lambda)}{\lambda -z},
$$
the ST of the empirical eigenvalue distribution in \eqref{eq:measure} can be written as:
\begin{equation}
\hat{m}_N(z)=\frac{1}{N}\sum_{k=1}^N \frac{1}{\lambda_k-z}.
\label{eq:m_Nz}
\end{equation}
Denote by  ${\bf Q}_N(z)=\left(\frac{1}{n}\bSigma_n\bSigma_n^*-z {\bf I}_N\right)^{-1}$. In the parlance of random matrix theory,  ${\bf Q}_N(z)$ is referred to as the resolvent matrix. From \eqref{eq:m_Nz}, one can easily see that:
\begin{equation}
\hat{m}_N(z)=\frac{1}{N}\tr {\bf Q}_N(z).
\label{eq:Q_N_m_N}
\end{equation}
Relation \eqref{eq:Q_N_m_N} clearly establishes the link between the resolvent matrix and the ST of the empirical eigenvalue distribution $\hat{\mu}_N$. It is a fundamental equation that accounts for the key role played by  the resolvent matrix in the theory of random matrices. 
%Denote by ${\bf Q}_N(z)=\left(\frac{1}{n}\bSigma_n\bSigma_n^*-z I_N\right)^{-1}$ the resolvent matrix corresponding to $\frac{1}{n}\bSigma_n\bSigma_n^*$. Then, the trace of ${\bf Q}_N(z)$ is nothing else than the Stieltjes transform of the empirical eigenvalue distribution $\hat{\mu}_N$. 
%This matrix plays a key role in random matrix theory. 
As a matter of fact, the study of the asymptotic behaviour of the resolvent matrix has provided an important load of new results concerning different statistical models \cite{HAC07,HAC06}. The model of generally correlated random matrices has recently been studied in \cite{WAG10}, where it has been proven that the ST of the empirical eigenvalue distribution converges almost surely to a deterministic function which is the ST of some probability distribution. 
%Of an increasing interest in current wireless applications, is the model of general correlated random matrices. Preliminary results in \cite{WAG11} establishes that   the Stieltjes transform of the empirical eigenvalue distribution  converge in distribution to a deterministic probability probability distribution $\mu_N$. However, this result is not sufficient to draw conclusions about the localization of the eigenvalues of $\frac{1}{n}\bSigma_n\bSigma_n^*$. If $\mathcal{S}_N$ denotes the support of $\mu_N$, one would imagine that in the asymptotic regime defined by Assumption \ref{ass:regime}, all the eigenvalues of  $\frac{1}{n}\bSigma_n\bSigma_n^*$ belong almost surely $\mathcal{S}_N$. This result has been first proven by Silverstein etal for different models from the one considered here \cite{SIL98,SIL12}. Very recently, a new approach based on Gaussian calculus has been proposed by Loubaton etal \cite{LOU10}. Although highly technical, this approach provides the key steps that can be used to establish similar results for more involved models. In this paper, we adapt the ideas of \cite{LOU10} to the general model presented above. 
%However, in order to characterize the behaviour of the spectrum of   $\bSigma_n\bSigma_n^*$ in the Gaussian setting, these results have to be further refined. To this end, in a similar way to \cite{LOU10}, we will extensively resort to  tools based on Gaussian calculus.
 More formally, it is well known from \cite{WAG10}, that it exists a sequence of deterministic measures $\mu_N$ such that $\hat{\mu}_N-\mu_N$ converges weakly to zero almost surely. Measure $\mu_N$ is characterized through its  ST $m_N(z)$ which is given by:
$$
m_N(z)=\frac{1}{N}\tr \left(\frac{1}{n}\sum_{i=1}^n\frac{\bOmega_i}{1+\delta_i(z)} -zI_N\right)^{-1},
$$
where $\delta_1,\cdots,\delta_n$ form the unique solutions that are ST of non-negative finite measure  of the following system of equations:
$$
\delta_i(z)=\frac{1}{n}\tr \bOmega_i \left(\frac{1}{n}\sum_{j=1}^n \frac{\bOmega_j}{1+\delta_j(z)} -z I_N\right)^{-1}
$$
for each $z\in \mathbb{C}\backslash \mathbb{R}^{+}$.

 In the following, we denote  by $\bT_N$, the matrix:
$$
\bT_N(z)=\left(\frac{1}{n}\sum_{i=1}^n\frac{\bOmega_i}{1+\delta_i(z)} -zI_N\right)^{-1},
$$
and
$$
m_N(z)=\frac{1}{N}\tr {\bf T}_N(z).
$$
 As $\hat{\mu}_N-\mu_N $ converge to zero weakly almost surely, we have:
$$
\hat{m}_N(z)-m_N(z)\stackrel{a.s.}{\to} 0
$$ 
for each $z\in\mathbb{C}\backslash\mathbb{R}^{+}$. %where $\hat{m}_N$ is the Stieltjes transform of $\hat{\mu}_N$ given also by:
%$$
%\hat{m}_N(z)=\frac{1}{N}\tr {\bf Q}_N(z).
%$$

%The proof of many of the most known results have been based on a deep study of its behaviour in the asymptotic regime \cite{}. 
%Indeed, this matrix will play a key role in the proof of our main results. 
%\subsection{Review of most known results about the behaviour of $\Sigma_n\Sigma_n^*$}
%The asymptotic behaviour of general correlated random matrices has been studied in \cite{WAG10}. In particular,
\section{Main results}
In this paper, we prove that under Assumptions \ref{ass:regime}-\ref{ass:boundedness}, the smallest eigenvalue of the Gram matrix $\frac{1}{n}\bSigma_n\bSigma_n^*$ stays away zero almost surely for $N$ large enough. This in particular implies, that for some $\epsilon >0$, $\hat{\mu}_N\left[0,\epsilon\right]=0$ for $N$ large enough. Since $\hat{\mu}_N-\mu_N$ converges weakly to zero, it is not difficult to convince oneself that one needs to start by showing that the support $\mathcal{S}_N$ of $\mu_N$ does not contain $0$. In particular, we prove the following result: 
%The almost sure convergence of $\hat{\mu}_N-\mu_N$ towards zero does not guarantee that the eigenvalues of  $\frac{1}{n}\bSigma_n\bSigma_n^*$ belong almost surely to the support $\mathcal{S}_N$ of $\mu_N$. This weak result need to be further refined in order to be able to locate the eigenvalues of $\frac{1}{n}\bSigma_n\bSigma_n^*$ with respect to $\mathcal{S}_N$ for $N$ large enough. 
%In this paper, we prove the following theorems:
\begin{theorem}
Under Assumption \ref{ass:regime} and \ref{ass:boundedness}, $0\notin \mathcal{S}_N$. In particular, there exists $\epsilon >0$ such that:
$$
\left[0,\epsilon\right]\cap \mathcal{S}_N=\emptyset.
$$
\label{th:zero}
\end{theorem}
To avoid disrupting the flow of the article, the proof of Theorem \ref{th:zero} is deferred to Appendix \ref{app:th_zero}.

Theorem \ref{th:zero} ensures that  $0$ does not belong to the support of the deterministic measure $\mu_N$. To conclude, it suffices to supplement this result with a second one, which establishes that almost surely, there is no eigenvalue of $\frac{1}{n}\bSigma_n\bSigma_n^*$ that goes outside the support of $\mathcal{S}_N$. This kind of result has already been shown to hold for other statistical models, by  either using  properties of the ST and bounds on the moments of martingale difference sequences \cite{SIL98,Paul-09,SIL12} or resorting to tools based on Gaussian calculus\cite{LOU10}. Since we assume in this paper that $\bSigma_n$ has Gaussian entries, we rather build on the method of \cite{LOU10} which also originates from some of the ideas of \cite{Capitaine-10}.  In particular, we establish the following result:
%This kind of result has already been established for other statistical models \cite{SIL98,Paul-09,SIL12,LOU10}.
%This will be shown in the following theorem:
%To conclude, we need to establish that there is no eigenvalue of $\frac{1}{n}\bSigma_n\bSigma_n^{\mbox{\tiny H}}$ that goes outside the support $\mathcal{S}_N$. This will be shown in the following theorem: 
\begin{theorem}
Assume that there exists a positive quantity $\epsilon >0$ and two real values $a,b\in\mathbb{R}$ such that for all $N$ large enough:
$$
\left]a-\epsilon,b+\epsilon\right[\cap \mathcal{S}_N =\emptyset
$$
Then, with probability one, no eigenvalue of $\frac{1}{n}\bSigma_n\bSigma_n^*$ appears in $\left[a,b\right]$ for all $N$ large enough. 
\label{th:no_eigenvalue}
\end{theorem}
\begin{proof}
The following proposition will be crucial in order to prove Theorem \ref{th:no_eigenvalue}. It merely quantifies the error that we incur by replacing $\mathbb{E}\frac{1}{N}\tr {\bf Q}(z)$ by $\frac{1}{N}\tr {\bf T}_N(z)$. The proof is quite demanding and heavily relies on Gaussian calculus tools.  It will be detailed in the corpus of the paper, namely in section \ref{sec:technical}, since we believe that some intermediate results be of independent interest.
%We expose in section \ref{sec:technical} the details of the proof.
%The proof of Theorem \ref{th:no_eigenvalue} relies on the following result, whose proof is quite demanding.  % and is deferred to section \ref{sec:technical}.
\begin{proposition}
$\forall z\in\mathbb{C}\backslash \mathbb{R}_{+}$, we have for $N$ large enough, 
$$
\mathbb{E}\left[\frac{1}{N}\tr{\bf {\bf Q}}(z)\right] =\frac{1}{N}\tr {\bf T}_N(z) +\frac{1}{N^2}\chi_N(z)
$$
with $\chi$ is analytic on $\mathbb{C}\backslash\mathbb{R}_{+}$ and satisfies:
\begin{equation}
\left|\chi_N(z)\right| \leq K\left(\left|z\right|+C\right)^k P\left(\left|\Im z\right|^{-1}\right)
\label{eq:technical}
\end{equation}
for each $z\in\mathbb{C}_{+}$ where $C,K$ are constants, $k$ is an integer independent of $N$ and $P$ is a polynomial with positive coefficients independent of $N$. 
\label{prop:technical}
\end{proposition}
Proposition \ref{prop:technical} will essentially serve to provide asymptotic approximates of linear statistics of the eigenvalues of the Gram matrix.  In fact, with the help of proposition \ref{prop:technical}, we prove the following result:
\begin{lemma}
Let $\phi$ be a compactly supported real-valued smooth function defined on $\mathbb{R}$, i.e, $\phi\in\mathcal{C}_c^{\infty}(\mathbb{R},\mathbb{R})$. Then \footnote{If ${\bf A}=\sum_{i=1}^N \lambda_i{\bf u}_i{\bf u}_i^{\mbox{\tiny H}}$ is an eigenvalue decomposition of ${\bf A}$, then $\phi({\bf A})=\sum_{i=1}^N \phi(\lambda_i){\bf u}_i{\bf u}_i^{\mbox{\tiny H}}$.},
\begin{equation}
\mathbb{E}\left[\phi\left(\frac{1}{N}\bSigma_n\bSigma_n^{*}\right)\right] -\int_{\mathcal{S}_N} \phi(\lambda) d\mu_N(\lambda)=\mathcal{O}\left(\frac{1}{N^2}\right).
\label{eq:difference}
\end{equation}
\end{lemma}
\begin{proof}
The proof is built around the use of the inversion lemma of ST. Recall that if $m$ is the ST of some finite measure $\mu$, then for any continuous real function $\phi$ with compact support in $\mathbb{R}$
$$
\int_{\mathbb{R}}\phi(\lambda)\mu(d\lambda)=\frac{1}{\pi} \Im\left(\lim_{y\downarrow 0}\int_{\mathbb{R}} \phi(x)m(x+\imath y)dx\right).
$$
We therefore have:
\begin{small}
\begin{align*}
&\mathbb{E}\left[\frac{1}{N}\tr\phi(\frac{1}{n}\bSigma_n\bSigma_n^{\mbox{\tiny H}})\right]=\frac{1}{\pi} \Im\left(\lim_{y\downarrow 0}\int_{\mathbb{R}} \phi(x)\mathbb{E}\left[\frac{1}{N}\tr {\bf Q}\left(x+\imath y\right)\right]dx\right) \\
&\int_{\mathcal{S}_N} \phi(\lambda)d\mu_N(\lambda)=\frac{1}{\pi} \Im\left(\lim_{y\downarrow 0}\int_{\mathbb{R}} \phi(x)\mathbb{E}\left[\frac{1}{N}\tr {\bf T}_N\left(x+\imath y\right)\right]dx\right).
\end{align*}
\end{small}
By proposition \ref{prop:technical}, we get:
\begin{align*}
&\mathbb{E}\left[\frac{1}{N}\tr\phi(\frac{1}{n}\bSigma_n\bSigma_n^{\mbox{\tiny H}})\right] - \int_{\mathcal{S}_N} \phi(\lambda)d\mu_N(\lambda) \\
&=\frac{1}{N^2}\frac{1}{\pi} \lim_{y\downarrow 0} \Im\left[\int_{\mathbb{R}_{+}} \phi(x)\chi_N(x+\imath y) dx\right].
\end{align*}
Since the function $\chi_N(z)$ satisfies \eqref{eq:technical}, Theorem 6.2 in \cite{haagerup} implies that:
$$
\lim\sup_{y\downarrow 0}\left|\int_{\mathbb{R}} \phi(x)\chi_N(x+\imath y)dx\right| \leq C<+\infty.
$$
where $C$ is a constant independent of $N$, thereby establishing \eqref{eq:difference}.
\end{proof}
We return now to the proof of Theorem \ref{th:no_eigenvalue}. With the above results at hand, Theorem \ref{th:no_eigenvalue} can be shown along the same lines as the proof of Theorem 3 in \cite{LOU10}. The details are  provided in the sequel for sake of completeness. 
Consider $\psi\in\mathcal{C}_c^{\infty}(\mathbb{R},\mathbb{R})$ satisfying $0\leq \psi \leq 1$ and:
$$
\psi(\lambda)=\left\{\begin{array}{lll}
1 & \textnormal{for} & \lambda\in\left[a,b\right] \\
0& \textnormal{for} & \lambda\in\mathbb{R}\backslash\left]a-\epsilon,b+\epsilon\right[.
\end{array}\right.
$$ 
For $N$ large enough, function $\psi$ is zero in the support $\mathcal{S}_N$. Therefore,
$$
\mathbb{E}\left[\frac{1}{N}\psi\left(\frac{1}{n}\bSigma_n\bSigma_n^{\mbox{\tiny H}}\right)\right] =\mathcal{O}\left(\frac{1}{N^2}\right).
$$
We need also to prove that the variance of $\frac{1}{N}\psi(\frac{1}{n}\bSigma_n\bSigma_n^{\mbox{\tiny H}})$ is of order $\frac{1}{N^4}$:
\begin{equation}
{\rm var}\left[\frac{1}{N}\psi\left(\frac{1}{n}\bSigma_n\bSigma_n^{\mbox{\tiny H}}\right)\right]=\mathcal{O}\left(\frac{1}{N^4}\right).
\label{eq:variance}
\end{equation}
To establish \eqref{eq:variance}, it suffices to resort to the Nash-Poincar\'e inequality which is stated in Lemma \ref{lemma:nash} of the next section.  Applying Lemma \ref{lemma:nash}, we obtain:
\begin{align}
&{\rm var}\left(\frac{1}{N}\tr \psi\left(\frac{1}{n}\bSigma_n\bSigma_n^{\mbox{\tiny H}}\right)\right)\leq\nonumber\\
&\sum_{k=1}^n \sum_{s=1}^N\sum_{r=1}^N \frac{\partial \frac{1}{N}\tr \psi\left(\frac{1}{n}\bSigma_n\bSigma_n^{\mbox{\tiny H}}\right)}{\partial \xi_{s,k}} \left[\bOmega_k\right]_{s,r}\frac{\left[\partial \frac{1}{N}\tr \psi\left(\frac{1}{n}\bSigma_n\bSigma_n^{\mbox{\tiny H}}\right)\right]^*}{\partial \xi_{r,k}}\nonumber\\
&+\sum_{k=1}^n \sum_{s=1}^N\sum_{r=1}^N \frac{\partial \frac{1}{N}\tr \psi\left(\frac{1}{n}\bSigma_n\bSigma_n^{\mbox{\tiny H}}\right)}{\partial \xi_{s,k}^*} \left[\bOmega_k\right]_{s,r}\frac{\left[\partial \frac{1}{N}\tr \psi\left(\frac{1}{n}\bSigma_n\bSigma_n^{\mbox{\tiny H}}\right)\right]^*}{\partial \xi_{r,k}^*}.\label{eq:nash}
\end{align}
By Lemma 4.6 in \cite{haagerup}, we have:
\begin{align}
\frac{\partial \left[\frac{1}{N}\tr \psi\left(\frac{1}{n}\bSigma_n\bSigma_n^{\mbox{\tiny H}}\right)\right]}{\partial \xi_{s,k}}&=\left[\frac{1}{Nn}\bSigma_n^{\mbox{\tiny H}}\psi^{'}\left(\frac{1}{n}\bSigma_n\bSigma_n^{\mbox{\tiny H}}\right)\right]_{k,s} \label{eq:ineq_1}\\
\frac{\partial \left[\frac{1}{N}\tr \psi\left(\frac{1}{n}\bSigma_n\bSigma_n^{\mbox{\tiny H}}\right)\right]}{\partial \xi_{s,k}^*}&=\left[\frac{1}{Nn}\psi^{'}\left(\frac{1}{n}\bSigma_n\bSigma_n^{\mbox{\tiny H}}\right)\bSigma_n\right]_{s,k}. \label{eq:ineq_2}
\end{align}
Plugging \eqref{eq:ineq_1} and \eqref{eq:ineq_2} into \eqref{eq:nash}, we get:
\begin{small}
\begin{align*}
&{\rm var}\left[\frac{1}{N}\psi\left(\frac{1}{n}\bSigma_n\bSigma_n^{\mbox{\tiny H}}\right)\right] \\
&\leq \sum_{k=1}^n \frac{2}{N^2n^2}\mathbb{E}\left[\tr \left(\bSigma_n\bSigma_n^{\mbox{\tiny H}} \psi^{'}\left(\frac{1}{n}\bSigma_n\bSigma_n^{\mbox{\tiny H}}\right) \bOmega_k \psi^{'}\left(\frac{1}{n}\bSigma_n\bSigma_n^{\mbox{\tiny H}}\right)\right)\right]\\
&\stackrel{(a)}{\leq} w_{\rm max} \sum_{k=1}^n \frac{2}{N^2n^2}\mathbb{E}\left[\tr \left(  \psi^{'}\left(\frac{1}{n}\bSigma_n\bSigma_n^{\mbox{\tiny H}}\right)\bSigma_n\bSigma_n^{\mbox{\tiny H}} \psi^{'}\left(\frac{1}{n}\bSigma_n\bSigma_n^{\mbox{\tiny H}}\right) \right)\right],
\end{align*}
\end{small}
where $(a)$ follows from the fact that $\tr{\bf A}{\bf B} \leq \|{\bf A}\|\tr {\bf B}$ for ${\bf A}$ hermitian and   ${\bf B}$ positive definite matrix. 
Consider $h:\lambda\mapsto \lambda \left|\psi^{'}(\lambda)\right|^2$. Clearly $h$ belongs to $\mathcal{C}_c^{\infty}(\mathbb{R},\mathbb{R})$. We therefore have:
\begin{align*}
&\mathbb{E}\left[\frac{1}{n}\tr \left(  \psi^{'}\left(\frac{1}{n}\bSigma_n\bSigma_n^{\mbox{\tiny H}}\right)\bSigma_n\bSigma_n^{\mbox{\tiny H}} \psi^{'}\left(\frac{1}{n}\bSigma_n\bSigma_n^{\mbox{\tiny H}}\right) \right)\right] \\
&=\int_{\mathcal{S}_N} h(\lambda)d\mu_N(\lambda)+\mathcal{O}\left(\frac{1}{N^2}\right).
\end{align*}
It is clear that for $N$ large enough, $\int_{\mathcal{S}_N} h(\lambda)d\mu_N(\lambda)=0$, thus proving:
$$
{\rm var}\left(\frac{1}{N}\psi\left(\frac{1}{n}\bSigma_n\bSigma_n^{\mbox{\tiny H}}\right)\right) =\mathcal{O}\left(\frac{1}{N^2}\right). 
$$
Applying the classical Markov inequality, we obtain:
\begin{align*}
&\mathbb{P}\left(\frac{1}{N}\tr \psi\left(\frac{1}{n}\bSigma_n\bSigma_n^{\mbox{\tiny H}}\right)\right) \leq N^{8/3} \mathbb{E}\left[\left|\frac{1}{N}\tr \psi\left(\frac{1}{n}\bSigma_n\bSigma_n^{\mbox{\tiny H}}\right)\right|^2\right] \\
&= N^{8/3}\left(\left|\mathbb{E}\left[\frac{1}{N}\tr \psi\left(\frac{1}{n}\bSigma_n\bSigma_n^{\mbox{\tiny H}}\right)\right]\right|^2\right.\\
&\left.+{\rm var}\left(\frac{1}{N}\tr \psi\left(\frac{1}{n}\bSigma_n\bSigma_n^{\mbox{\tiny H}}\right)\right)\right) \\
&=\mathcal{O}\left(\frac{1}{N^{4/3}}\right).
\end{align*} 
Thus, by Borel-Cantelli lemma, for $N$ large enough,
$$
\frac{1}{N}\tr \psi\left(\frac{1}{n}\bSigma_n\bSigma_n^{\mbox{\tiny H}}\right) \leq \frac{1}{N^{4/3}},
$$
or equivalently,
$$
\tr \psi\left(\frac{1}{n}\bSigma_n\bSigma_n^{\mbox{\tiny H}}\right)\leq \frac{1}{N^{1/3}}
$$
By definition of function $\psi$, the number of eigenvalues of the Gram matrix $\frac{1}{n}\bSigma_n\bSigma_n^{\mbox{\tiny H}}$ that lies in the in the interval $\left[a,b\right]$ is upper-bounded by $\tr \psi\left(\frac{1}{n}\bSigma_n\bSigma_n^{\mbox{\tiny H}}\right)$, and is therefore less than $\frac{1}{N^{1/3}}$ with probability $1$. Since this number has to be an integer, we deduce that it is zero for $N$ large enough. As a consequence, there is no eigenvalue in $\left[a,b\right]$ for $N$ large enough.

\end{proof}
%Theorem \ref{th:no_eigenvalue} ensures that almost surely, there is no eigenvalue of  $\frac{1}{n}\bSigma_n\bSigma_n^{\mbox{\tiny H}}$ that goes outside the support $\mathcal{S}_N$. To prove that the smallest eigenvalue of $\frac{1}{n}\bSigma_n\bSigma_n^{\mbox{\tiny H}}$ is almost surely bounded away zero, it suffices to establish that $0$ does not belong to $\mathcal{S}_N$. In particular, in this paper, we establish that:

Gathering the results of Theorem \ref{th:no_eigenvalue} and Theorem \ref{th:zero}, we get:
\begin{corollary}
Assume the setting of Theorem \ref{th:zero}. Then, for $N$ large enough, the smallest eigenvalue of $\bSigma_n\bSigma_n^*$ is bounded away from zero.
\end{corollary}

\section{Approximation rule}
\label{sec:technical}
This section aims at showing the approximation in proposition \ref{prop:technical} stating that:
$$
\mathbb{E}\left[\frac{1}{N}\tr {\bf Q}(z)\right]=\frac{1}{N}\tr {\bf T}_N(z)+\frac{1}{N^2}\chi_N(z)
$$
for $N$ large enough, where $\chi$ is analytic on $\mathbb{C}\backslash \mathbb{R}_{+}$ and satisfies inequality \eqref{eq:technical}.

As far as generally correlated Gaussian matrices are concerned, the convergence of $\frac{1}{N}\tr {\bf Q}_N(z)$ to $\frac{1}{N}\tr {\bf T}_N(z)$ has been shown to hold in the almost sure sense, \cite{WAG10}. This result directly implies that the empirical eigenvalue distribution converges weakly to a measure $\mu_N$ which is characterized by its stieltjes transform $m_N(z)=\frac{1}{N}\tr {\bf T}_N(z)$. Its importance lies in that it gives us insights on the proportion of eigenvalues falling in any interval. But, it does not rule out the possibility of a $o(n)$ proportion of eigenvalues lying outside the limiting support of $\mu_N$. 
As it has been shown above, a sufficient condition that can eliminate this possibility is constituted by the statement of proposition \ref{prop:technical}. This statement is already known to hold for other models, mainly the non-centered Gaussian model \cite{LOU10}. Its proof for the model of generally correlated Gaussian matrices has not been carried out, to the best of the authors' knowledge. 

While the proof of proposition \ref{prop:technical} relies on the standard use of Gaussian calculus tools, several adaptations to the specificity of the random matrix model are far from being immediate. To facilitate the understanding of the highly technical proof, we start by introducing the main key steps. In order to control the difference $\frac{1}{N}\mathbb{E}\tr{\bf Q}_N(z)-\frac{1}{N}\tr {\bf T}_N(z)$, we need to introduce, similar to previous works \cite{HAC06}, an intermediate deterministic matrix denoted by ${\bf R}_N(z)$ and which writes as:
$$
{\bf R}_N(z)=\left(\frac{1}{n}\sum_{k=1}^n \frac{\bOmega_k}{1+\alpha_k(z)}-z{\bf I}_N\right)^{-1},
$$
where $\alpha_k(z)=\frac{1}{n}\tr \bOmega_k\mathbb{E}{\bf Q}(z)$, $k=1,\cdots,n$.
With matrix ${\bf R}_N(z)$ at hand, we decompose the difference $\frac{1}{N}\mathbb{E}\tr{\bf Q}_N(z)-\frac{1}{N}\tr {\bf T}_N(z)$ as:
\begin{align*}
\frac{1}{N}\mathbb{E}\tr{\bf Q}_N(z)-\frac{1}{N}\tr {\bf T}_N(z)&=\frac{1}{N}\mathbb{E}\tr{\bf Q}_N(z)-\frac{1}{N}\tr {\bf R}_N(z) \\
&+\frac{1}{N}\tr{\bf R}_N(z)-\frac{1}{N}\tr {\bf T}_N(z)\\
&\triangleq \frac{1}{N^2}\chi_1(z)+\frac{1}{N^2}\chi_2(z).
\end{align*}
This decomposition is quite standard in random matrix theory. While the direct control of the difference $\frac{1}{N}\mathbb{E}\tr{\bf Q}_N(z)-\frac{1}{N}\tr {\bf T}_N(z)$ is complicated, much can be inferred from both differences $\frac{1}{N}\mathbb{E}\tr{\bf Q}_N(z)-\frac{1}{N}\tr {\bf R}_N(z) $ and $\frac{1}{N}\tr{\bf R}_N(z)-\frac{1}{N}\tr {\bf T}_N(z)$.
In order to prove proposition \ref{prop:technical}, it suffices to show that:
$$
\left|\chi_i(z)\right| \leq \left(|z|+C_i\right)^{k_i} P_i\left(\left|\Im z\right|^{-1}\right), i=1,2,
$$
where $C_i,i=1,2$ are positive constants,  $k_i,i=1,2$ are positive integers and $P_i,i=1,2$ are polynomial with positive coefficients independent of $N$.
In addition to ${\bf R}_N(z)$, we will need to introduce the following deterministic quantities:
\begin{align*}
\tilde{r}_i&=-\frac{1}{z(1+\alpha_i(z))}, i=1,\cdots,n \\
\tilde{\bf R}_N&={\rm diag}\left(\tilde{r}_1,\cdots,\tilde{r}_n\right).
\end{align*}
It can be easily shown along the same lines of Proposition 5.1 of \cite{HAC07} that matrix valued functions ${\bf R}_N(z)$ and $\tilde{\bf R}_N(z)$ are holomorphic in $\mathbb{C}\backslash{\mathbb{R}_+}$ and coincide with the Stieltjes transforms of positive matrix valued probability measures carried by $\mathbb{R}_{+}$, the mass of which are equal to ${\bf I}$. Their spectral norms are thus bounded by $\left|\Im z\right|^{-1}$. In particular, we have:
$$
{\rm max}\left(\|\tilde{\bf R}_N\|,\|{\bf R}_N\|\right) \leq \left|\Im z\right|^{-1}.
$$
With these quantities at hand, we are now in position to sequentially control the terms $\chi_1(z)$ and $\chi_2(z)$.
\subsection{Control of $\chi_1(z)$}
The control of $\chi_1(z)$ will extensively rely on the use of Gaussian calculus tools, namely the Integration by Part formulae and the Nash-Poincar\'e inequality. Before delving into the core of the proof, we shall recall these tools.
\begin{lemma}[Integration by Part Lemma]
Let ${\bf x}=\left[x_1,\cdots,x_N\right]^{\mbox{\tiny T}}$ a complex Gaussian vector such that $\mathbb{E}\left[{\bf x}\right]=0$, $\mathbb{E}\left[{\bf x}{\bf x}^{\mbox{\tiny T}}\right]=0$ and $\mathbb{E}\left[{\bf x}{\bf x}^*\right]={\bf R}$. If $\Gamma:{\bf x}\mapsto \Gamma({\bf x})$ is a $\mathcal{C}^{1}$ complex function polynomially bounded together with its derivatives, then:
$$
\mathbb{E}\left[x_p\Gamma(x)\right]=\sum_{m=1}^N \left[{\bf R}\right]_{p,m}\mathbb{E}\left[\frac{\partial \Gamma({\bf x})}{\partial x_m^*}\right]
$$
\label{lemma:integration_part}
\end{lemma}
\begin{lemma}[Nash-Poincar\'e Inequality]
Let ${\bf x}=\left[x_1,\cdots,x_N\right]^{\mbox{\tiny T}}$ a complex Gaussian vector such that $\mathbb{E}\left[{\bf x}\right]=0$, $\mathbb{E}\left[{\bf x}{\bf x}^{\mbox{\tiny T}}\right]=0$ and $\mathbb{E}\left[{\bf x}{\bf x}^*\right]={\bf R}$. If $\Gamma:{\bf x}\mapsto \Gamma({\bf x})$ is a $\mathcal{C}^{1}$ complex function polynomially bounded together with its derivatives, then, noting $\nabla_{x}\Gamma=\left[\frac{\partial \Gamma}{\partial x_1},\cdots,\frac{\partial \Gamma}{\partial x_M}\right]^{\mbox{\tiny T}}$ and  $\nabla_{x^*}\Gamma=\left[\frac{\partial \Gamma}{\partial x_1^*},\cdots,\frac{\partial \Gamma}{\partial x_M^*}\right]^{\mbox{\tiny T}}$,
\begin{align*}
	{\rm var}\left(\Gamma(x)\right) &\leq \mathbb{E}\left[\nabla_{\bf x}\Gamma(x)^{\mbox{\tiny T}}{\bf R} \ \ {\left(\nabla_{\bf x}\Gamma(x)\right)}^*\right] \\
	&+\mathbb{E}\left[\left(\nabla_{{\bf x}^*}\Gamma(x)\right)^{*}{\bf R} \nabla_{{\bf x}^*}\Gamma(x)\right].
\end{align*}
\label{lemma:nash}
\end{lemma}
Applying Lemma \ref{lemma:nash}, we will thus get:
\begin{multline}
{\rm var}\left(\Gamma(\xi_1,\cdots,\xi_n)\right)\leq 
\sum_{k=1}^n\sum_{s=1}^N\sum_{r=1}^N \mathbb{E}\left[\frac{\partial \Gamma}{\partial \xi_{s,k}}\left[\Omega_k\right]_{s,r}\frac{\partial {\Gamma}^*}{\partial \xi_{r,k}}\right] \\
+\sum_{k=1}^n \sum_{s=1}^N\sum_{r=1}^N\mathbb{E}\left[ \frac{\partial {\Gamma}^*}{\partial \xi_{s,k}^*} \left[\Omega_k\right]_{s,r}\frac{\partial \Gamma}{\partial \xi_{r,k}^*}\right].
\label{eq:variance}
\end{multline}
The application of these tools will require us to compute differentials of the resolvent matrix with respect to the entries of $\bSigma_n$. In particular, we will need in the sequel, the following differentiation formulas:
\begin{align}
\frac{\partial \left[{\bf Q}\right]_{\ell,p}}{\partial {\xi_{m,k}^*}}&=-\frac{1}{n}\frac{\left[{\bf Q}\partial \bSigma_n\bSigma_n^* {\bf Q}\right]_{\ell,p}}{\partial {\xi_{m,k}^*}}\nonumber\\
&=-\frac{1}{n}\left[{\bf Q}\bxi_ke_m^{\mbox{\tiny T}}{\bf Q}\right]_{\ell,p} \nonumber\\
&=-\frac{1}{n}\left[{\bf Q}\bxi_k\right]_\ell\left[{\bf Q}\right]_{m,p}. \label{eq:diff_1}
\end{align}
Moreover, we also have:
\begin{align}
\frac{\partial \left[{\bf Q}\right]_{\ell,p}}{\partial \xi_{s,k}} &=-\frac{1}{n}\left[{\bf Q}\right]_{\ell,s}\left[\bxi_k^*{\bf Q}\right]_p.
\label{eq:diff_2}
\end{align}
The use of the integration by part lemma along with the above differential formulae will allow us to establish the following lemma:
\begin{lemma}
Let $\beta_i,i=1,\cdots,n$ be given by $\beta_i=\frac{1}{n}\tr \bOmega_i{\bf Q}(z)$.
For each $z\in\mathbb{C}_{+}$ and any deterministic matrix ${\bf A}$, it holds that:
$$
\mathbb{E}\tr {\bf A}{\bf Q}(z)=\tr {\bf A}{\bf R}(z)-z\mathbb{E}\tr \frac{{\bf A}{\bf Q}\bSigma_n\tilde{\bf R}{\bf B}\bSigma_n^*{\bf R}}{n}
$$
where ${\bf B}={\rm diag}\left(\stackrel{o}{\beta}_1,\cdots,\stackrel{o}{\beta}_n\right)$ with$$
\stackrel{o}{\beta}_i =\beta_i-\alpha_i.
$$
\label{lemma:AQ}
\end{lemma}
\begin{proof}
From the identity:
$$
{\bf Q}\left(\frac{1}{n}\bSigma_n\bSigma_n^*-z{\bf I}_N\right) ={\bf I}_N
$$
we have:
\begin{align}
z\mathbb{E}\left[{\bf Q}\right]_{p,q}&= \mathbb{E}\left[{\bf Q}\frac{\bSigma_n\bSigma_n^*}{n}\right]_{p,q}-\delta_{p,q}\label{eq:tobe_plug}\\
&=\sum_{i=1}^N \sum_{j=1}^n \frac{1}{n}\mathbb{E}\left[{\bf Q}_{p,i}\xi_{i,j}\xi_{q,j}^*\right] -\delta_{p,q}.\nonumber
\end{align}
Using the integration by parts formula in Lemma \ref{lemma:integration_part}, we have:
\begin{align*}
&\mathbb{E}\left[{\bf Q}_{p,i}\xi_{i,j}\xi_{q,j}^*\right]=\sum_{m=1}^N \mathbb{E}\left[\left[\bOmega_j\right]_{i,m}\frac{\partial \xi_{q,j}^*\left[{\bf Q}\right]_{p,i}}{\partial \xi_{m,j}^*}\right]\\
&=\sum_{m=1}^N \left[\bOmega_j\right]_{i,m}\delta_{m,q}\mathbb{E}\left[{\bf Q}\right]_{p,i} \\
&- \sum_{m=1}^N \left[\bOmega_j\right]_{i,m} \frac{1}{n}\mathbb{E}\left[\xi_{q,j}^*\left[{\bf Q}\bxi_j\right]_p \left[{\bf Q}\right]_{m,i} \right].
\end{align*}
Summing the above equality over $i$, we obtain:
$$
\mathbb{E}\left[\left[{\bf Q}\bxi_j\right]_p\xi_{q,j}^*\right]=\mathbb{E}\left[{\bf Q}\bOmega_j\right]_{p,q} - \mathbb{E}\left[\beta_j \left[{\bf Q}\bxi_j\right]_p \xi_{q,j}^*\right]
$$
Plugging $\stackrel{o}{\beta}_j={\beta}_j-\alpha_j$ into the above equality, we get:
\begin{align*}
\mathbb{E}\left[\left[{\bf Q}\bxi_j\right]_p\xi_{q,j}^*\right]&=\mathbb{E}\left[{\bf Q}\bOmega_j\right]_{p,q} -\alpha_j \mathbb{E}\left[\xi_{q,j}^*\left[{\bf Q}\bxi_j\right]_p\right]\\
&-\mathbb{E}\left[\stackrel{o}{\beta}_j\xi_{q,j}^*\left[{\bf Q}\bxi_j\right]_p\right]
\end{align*}
Hence:
$$
\mathbb{E}\left[\left[{\bf Q}\bxi_j\right]_p\xi_{q,j}^*\right]=\mathbb{E}\left[\frac{\left[{\bf Q}\bOmega_j\right]_{p,q}}{(1+\alpha_j)}\right] -\mathbb{E}\left[\frac{\stackrel{o}{\beta}_j\xi_{q,j}^*\left[{\bf Q}\bxi_j\right]_p}{(1+\alpha_j)}\right]
$$
Summing over $j$, we finally get:
\begin{align*}
\mathbb{E}\left[\frac{{\bf Q}\bSigma_n\bSigma_n^*}{n}\right]_{p,q}&=\mathbb{E}\left[{\bf Q}\frac{1}{n}\sum_{j=1}^n \frac{\bOmega_j}{(1+\alpha_j)}\right]_{p,q}\\
&+z\mathbb{E}\left[\frac{{\bf Q}\bSigma_n\tilde{\bf R}{\bf B}\bSigma_n^*}{n}\right]_{p,q}
\end{align*}
Plugging the above equality into \eqref{eq:tobe_plug}, we thus get:
\begin{align*}
\mathbb{E}\left[z{\bf Q}\right]_{p,q}&=\mathbb{E}\left[{\bf Q}\frac{1}{n}\sum_{j=1}^n \frac{\bOmega_j}{(1+\alpha_j)}\right]_{p,q}-\left[{\bf I}_N\right]_{p,q}\\
&+z\mathbb{E}\left[\frac{{\bf Q}\bSigma_n\tilde{\bf R}{\bf B}\bSigma_n^*}{n}\right]_{p,q}
\end{align*}
Therefore,
$$
\mathbb{E}\left[{\bf Q}{\bf R}^{-1}\right]_{p,q}=\left[{\bf I}_N\right]_{p,q}-z\mathbb{E}\left[\frac{{\bf Q}\bSigma_n\tilde{\bf R}{\bf B}\bSigma_n^*}{n}\right]_{p,q}
$$
thereby proving that:
$$
\mathbb{E}{\bf Q}{\bf R}^{-1}={\bf I}_N - z\mathbb{E}\left[\frac{{\bf Q}\bSigma_n\tilde{\bf R}{\bf B}\bSigma_n^*}{n}\right].
$$
As a consequence:
$$
\mathbb{E}\tr {\bf A}{\bf Q}=\tr {\bf A}{\bf R} -\frac{z}{n}\tr\mathbb{E}\left[{\bf A}{\bf Q}\bSigma_n\tilde{\bf R}{\bf B}\bSigma_n^*{\bf R}\right].
$$
%\begin{align*}
%\mathbb{E}\tr {\bf A Q}&=\sum_{q=1}^N\sum_{p=1}^N \mathbb{E}\left[\left[{\bf A}\right]_{q,p}\left[{\bf Q}\right]_{p,q}\right] 
%&=\sum_{q=1}^N \sum_{p=1}^N \mathbb{E}\left[\frac{1}{z}{\bf A}_{q,p}\left[{\bf Q}\frac{\bSigma_n\bSigma_n^*}{n}\right]_{p,q}-\frac{1}{z}\delta_{p,q}\right]\\
%&=\mathbb{E}\left[\left[{\bf A}\right]_{q,p}\left[{\bf Q}\frac{1}{n}\sum_{j=1}^n\frac{\bOmega_j}{z(1+\alpha_j)}\right]_{p,q}\right] \\
%&-\frac{1}{z}\left[{\bf A}\right]_{q,p}\delta_{p,q}+\left[{\bf A}\right]_{q,p}\mathbb{E}\left[\frac{{\bf Q}\bSigma_n\tilde{\bf R}{\bf B}\bSigma_n^*}{n}\right]_{p,q}
%\end{align*}
%where:
%\begin{align*}
%\mathbb{E}\left[\left[{\bf A}\right]_{q,p}{\bf Q}_{}\right]
%\end{align*}
\end{proof}
From Lemma \ref{lemma:AQ}, it appears that the control of $\chi_1$ amounts to showing that:
$$
z\Gamma\triangleq z\mathbb{E}\left[\tr {\bf Q}\bSigma_n\tilde{\bf R}{\bf B}\bSigma_n^*{\bf R}\right] \leq \frac{1}{n}\left(|z|+C_1\right)^{k_1} P_1\left(\left|\Im z\right|^{-1}\right)
$$
with $C_1$, $k_1$ and $P_1$ verifying the conditions of proposition \ref{prop:technical}.
The proof relies on the use of the Nash-poincar\'e inequality. But before that, we need to further workout quantity $\Gamma$ by means of the Integration by Part formula. We first expand $\Gamma$ as:
\begin{align}
\Gamma&=\frac{1}{n}\sum_{p,q,m=1}^N \sum_{\ell=1}^n \mathbb{E}\left[\left[{\bf Q}\right]_{p,q}\xi_{q,\ell} \xi_{m,\ell}^*\stackrel{o}{\beta}_\ell \right]\left[{\bf R}\right]_{m,p}\tilde{r}_\ell\label{eq:gamma}
\end{align}
Using the integration by part formula, we have:
\begin{small}
\begin{align*}
&\mathbb{E}\left[\left[{\bf Q}\right]_{p,q}\xi_{q,\ell}\xi_{m,\ell}^*\stackrel{o}{\beta}_\ell\right]=\sum_{s=1}^N \left[\bOmega_\ell\right]_{q,s}\mathbb{E}\left[\frac{\partial \left[{\bf Q}\right]_{p,q}\xi_{m,\ell}^*\stackrel{o}{\beta}_\ell}{\partial \xi_{s,\ell}^*}\right]\\
&=\sum_{s=1}^{N}\left[\bOmega_\ell\right]_{q,s}\mathbb{E}\left[\left[{\bf Q}\right]_{p,q}\stackrel{o}{\beta}_\ell\right] \delta_{m,s}+\sum_{s=1}^N \left[\bOmega_\ell\right]_{q,s}\mathbb{E}\left[\left[{\bf Q}\right]_{p,q}\xi_{m,\ell}^*\frac{\partial \stackrel{o}{\beta}_\ell}{\partial \xi_{s,\ell}^*}\right]\\
& -\left[\bOmega_\ell\right]_{q,s}\frac{1}{n}\mathbb{E}\left[\left[{\bf Q}\bxi_\ell\right]_p\left[{\bf Q}\right]_{s,q}\xi_{m,\ell}^*\stackrel{o}{\beta}_\ell\right]\\
&=-\frac{1}{n}\mathbb{E}\left[\left[\bOmega_\ell {\bf Q}\right]_{q,q} \left[{\bf Q}\bxi_\ell\right]_p \xi_{m,\ell}^*\stackrel{o}{\beta}_\ell\right] +\left[\bOmega_\ell\right]_{q,m}\mathbb{E}\left[\left[{\bf Q}\right]_{p,q}\stackrel{o}{\beta}_\ell\right]\\
&+\sum_{s=1}^N \left[\bOmega_\ell\right]_{q,s}\mathbb{E}\left[\left[{\bf Q}\right]_{p,q}\xi_{m,\ell}^*\frac{\partial \stackrel{o}{\beta}_\ell}{\partial \xi_{s,\ell}^*}\right]
\end{align*}
 \end{small}
Summing the above equation over $q$, we get:
\begin{align*}
&\mathbb{E}\left[\left[{\bf Q}\bxi_\ell\right]_p\xi_{m,\ell}^*\stackrel{o}{\beta}_\ell\right]=-\mathbb{E}\left[\frac{1}{n}\tr\left(\bOmega_\ell {\bf Q}\right)\left[{\bf Q}\bxi_\ell\right]_p\xi_{m,\ell}^*\stackrel{o}{\beta}_\ell \right] \\
&+\mathbb{E}\left[\left[{\bf Q}\bOmega_\ell\right]_{p,m}\stackrel{o}{\beta}_\ell\right] +\sum_{q=1}^N\sum_{s=1}^N \left[\bOmega_\ell\right]_{q,s}\mathbb{E}\left[\left[{\bf Q}\right]_{p,q}\xi_{m,\ell}^*\frac{\partial \stackrel{o}{\beta}_\ell}{\partial \xi_{s,\ell}^*}\right]
\end{align*}
Writing $\frac{1}{n}\tr \bOmega_\ell {\bf Q}$ as $\stackrel{o}{\beta}_\ell+\alpha_\ell$ and using the same technique as in the proof of Lemma \ref{lemma:AQ}, we finally get:
\begin{align}
&\mathbb{E}\left[\left[{\bf Q}\bxi_\ell\right]_p\xi_{m,\ell}^*\stackrel{o}{\beta}_\ell\right]=z\tilde{r}_\ell \mathbb{E}\left[\left(\stackrel{o}{\beta}_\ell\right)^2\left[{\bf Q}\bxi_\ell\right]_p\xi_{m,\ell}^*\right]\nonumber\\
&-z\tilde{r}_\ell\mathbb{E}\left[\left[{\bf Q}\bOmega_\ell\right]_{p,m}\stackrel{o}{\beta}_\ell\right]-\sum_{s,q=1}^N z\tilde{r}_\ell\left[\bOmega_\ell\right]_{q,s}\mathbb{E}\left[\left[{\bf Q}\right]_{p,q}\xi_{m,\ell}^*\frac{\partial \stackrel{o}{\beta}_\ell}{\partial \xi_{s,\ell}^*}\right]\label{eq:final_b}
\end{align}
Plugging \eqref{eq:final_b} into \eqref{eq:gamma}, we finally obtain:
\begin{small}
\begin{align*}
\Gamma&=\frac{z}{n}\mathbb{E}\left[\tr\left( {\bf Q}\bSigma_n\tilde{\bf R}^2{\bf B}^2\bSigma_n^*{\bf R}\right)\right]-\frac{z}{n}\sum_{\ell=1}^n \mathbb{E}\left[\stackrel{o}{\beta}_\ell\tr\left({\bf Q}\bOmega_\ell{\bf R}\tilde{\bf R}^2\right)\right]\\
&-\frac{z}{n}\sum_{\ell=1}^n\sum_{s=1}^N \tilde{r}_\ell^2\mathbb{E}\left[\left[\bSigma_n^*{\bf R}{\bf Q}\bOmega_\ell\right]_{\ell,s}\frac{\partial \stackrel{o}{\beta}_{\ell}}{\partial \xi_{s,\ell}^*}\right]\\
&\triangleq \Delta_1 -\Delta_2 -\Delta_3.
\end{align*}
\end{small}
In the following we will prove that $\Delta_i$ satisfies:
$$
\Delta_i \leq \frac{K_i}{n}\left(|z|+\tilde{C}_i\right)^{\tilde{k}_i}\tilde{P}_i(|\Im z|^{-1})
$$
for some positive constant $\tilde{C}_i,K_i$, integer $k_i$ and polynomial $\tilde{P}_i$ independent of $N$.
This will be sufficient to control $\chi_1(z)$ since the underlying polynomials have positive coefficients.
Closer scrutiny of the expressions of $\Delta_i, i=1,2,3$, reveals that they make appear quantities of the form $\frac{1}{n}\tr {\bf A}{\bf Q}(z)$ with ${\bf A}$ is a some deterministic matrix. It is thus easy to convince oneself that controlling the variance of these terms is essential. This will be the goal of the following lemma whose proof is deferred to Appendix \ref{app:variance}:
\begin{lemma}
Let ${\bf A}$ be a $N\times N$ deterministic matrix. Then, we have for any $z\in\mathbb{C}_{+}$,
$$
{\rm var}\left(\frac{1}{n}\tr {\bf A}{\bf Q}(z)\right) \leq \frac{C}{n^2}\left\|{\bf A}\right\|^2\left(|z|+1\right)\left(\frac{1}{|\Im z|^4}+\frac{1}{|\Im z|^3}\right)
$$
where $C$, a positive constant and $P$, a polynomial with positive coefficients, are independent of $N$.
\label{lemma:var_A}
\end{lemma} 
With Lemma \ref{lemma:var_A} at hand, we are now in position to handle the terms $\Delta_i, i=1,2,3$. We start by controlling $\Delta_1$. For that, consider $\bSigma_{(i)}$ to be the matrix $\bSigma_n$ without its $i$-th column. Define ${\bf Q}_{(i)}$ the resolvent matrix given by:
$$
{\bf Q}_{(i)}=\left(\frac{1}{n}\bSigma_{(i)}\bSigma_{(i)}^*-z{\bf I}_N\right)^{-1}
$$
and $\beta_{i,(i)}=\frac{1}{n}\tr\bOmega_i{\bf Q}_{(i)}$. Let $\stackrel{o}{\beta}_{i,(i)}=\beta_{i,(i)}-\mathbb{E}\beta_{i,(i)}$ and ${\bf B}_{(i)}={\rm diag}\left(\stackrel{o}{\beta}_{1,(1)},\cdots,\stackrel{o}{\beta}_{n,(n)}\right)$. From the rank-one perturbation Lemma \cite[Lemma 2.6]{SIL95}, we obtain:
$$
\max_{1\leq i\leq n }\left|\stackrel{o}{\beta}_i-\stackrel{o}{\beta}_{i,(i)}\right| \leq \frac{2w_{\rm max}}{n\left|\Im z\right|}
$$
Decompose $\Delta_1$ as:
\begin{align*}
\Delta_1&=\frac{z}{n}\sum_{i=1}^n \mathbb{E}\left[\left(\left|\stackrel{o}{\beta}_i\right|^2-\left|\stackrel{o}{\beta}_{i,(i)}\right|^2\right)\left[\bSigma_n^*{\bf R}_N{\bf Q}\bSigma_n\tilde{\bf R}_N^2\right]_{i,i}\right] \\
&+\frac{z}{n}\sum_{i=1}^n \mathbb{E}\left[\left|\stackrel{o}{\beta}_{i,(i)}\right|^2\left[\bSigma_n^*{\bf R}_N{\bf Q}\bSigma_n\tilde{\bf R}_N^2\right]_{i,i}\right]\\
&\triangleq \Delta_{1,1}+\Delta_{1,2}.
\end{align*}
We start by dealing with $\Delta_{1,1}$. First, we need to bound the quantity $\left|\stackrel{o}{\beta}_i\right|^2-\left|\stackrel{o}{\beta}_{i,(i)}\right|^2$. We have:
\begin{align}
\left|\stackrel{o}{\beta}_i\right|^2-\left|\stackrel{o}{\beta}_{i,(i)}\right|^2&=\left(\left|\stackrel{o}{\beta}_i\right|-\left|\stackrel{o}{\beta}_{i,(i)}\right|\right)\left(\left|\stackrel{o}{\beta}_i\right|+\left|\stackrel{o}{\beta}_{i,(i)}\right|\right) \nonumber\\
&\leq \frac{2Nw_{\rm max}}{n\left|\Im z\right|} \left|\stackrel{o}{\beta}_i-\stackrel{o}{\beta}_{i,(i)}\right| \nonumber \\
&\leq \frac{4Nw_{\rm max}^2}{n^2\left|\Im z\right|^2}. \label{eq:Delta_11}
\end{align}
From \eqref{eq:Delta_11}, $\Delta_{1,1}$ can be bounded by:
$$
\Delta_{1,1} \leq \frac{|z|}{n^3}\frac{4N w_{\rm max}^2}{\left|\Im z\right|^2} \sum_{i=1}^n \mathbb{E}\left|\left[\bSigma_n^*{\bf R}_N{\bf Q}\bSigma_n\tilde{\bf R}_N^2\right]_{i,i}\right|.
$$
We need thus to bound $\mathbb{E}\left|\left[\frac{1}{n}\bSigma_n^*{\bf R}_N{\bf Q}\bSigma_n\tilde{\bf R}_N^2\right]_{i,i}\right|$. We have:
\begin{align*}
\mathbb{E}\left|\left[\frac{1}{n}\bSigma_n^*{\bf R}_N{\bf Q}\bSigma_n\tilde{\bf R}_N^2\right]_{i,i}\right|&=\mathbb{E}\left[\frac{1}{n}\bxi_i^*{\bf R}_N{\bf Q }\bxi_i\tilde{r}_i^2\right] \\
&\leq |\tilde{r}|_{i}^2 \mathbb{E}\left[\|{\bf R}_N{\bf Q}\| \frac{1}{n}\bxi_i^*\bxi_i\right] \\
&\leq \frac{1}{|\Im z|^4} \frac{1}{n}\tr \bOmega_i \\
&\leq \frac{Nw_{\rm max}}{n|\Im z|^4}
\end{align*}
%Using the relation:
%\begin{equation}
%{\bf Q}\bxi_i=\frac{{\bf Q}_{(i)} \bxi_i^*}{1+\frac{1}{n}\bxi_i^*{\bf Q}\bxi_i^*},
%\label{eq:inversion}
%\end{equation}
%we get:
%\begin{align*}
%\mathbb{E}\left[\frac{1}{n}\bSigma_n^*{\bf R}_N{\bf Q}\bSigma_n\tilde{\bf R}_N^2\right]_{i,i}& =\mathbb{E}\frac{1}{n}\frac{\bxi_i^*{\bf R}_N{\bf Q}_{(i)}\bxi_i\tilde{r}_i^2}{1+\frac{1}{n}\bxi_i^*{\bf Q}_{(i)}\bxi_i} \\
%&\leq |\tilde{r}_i|^2\mathbb{E}\frac{1}{n}\left|\bxi_i^*{\bf R}_N{\bf Q}_{(i)}\bxi_i\right| \\
%&\leq |\tilde{r}_i|^2 \mathbb{E} \left\|{\bf R}{\bf Q}_{(i)}\right\| \frac{1}{n}\bxi_i^*\bxi_i
%&\leq \frac{Nw_{\rm max}}{n|\Im z|^4}
%\end{align*}
and thus:
$$
\Delta_{1,1} \leq 4|z|\left(\limsup_N\frac{N}{n}\right)^2\frac{w_{\rm max}^3}{n|\Im z|^6}.
$$
We now move to the control of $\Delta_{1,2}$. First, write $\Delta_{1,2}$ as:
$$
\Delta_{1,2}=\frac{z}{n}\sum_{i=1}^n \mathbb{E}\left[\left|\stackrel{o}{\beta}_{i,(i)}\right|^2\left[\bxi_i^*{\bf R}{\bf Q}\bxi_i\tilde{r}_i^2\right]\right].
$$ 
Using the relation
\begin{equation}
{\bf Q}\bxi_i=\frac{{\bf Q}_{(i)} \bxi_i^*}{1+\frac{1}{n}\bxi_i^*{\bf Q}\bxi_i^*},
\label{eq:inversion}
\end{equation}
 we obtain:
\begin{align*}
\Delta_{1,2}&\leq \frac{|z|}{n}\sum_{i=1}^n \mathbb{E}\left[\left|\stackrel{o}{\beta}_{i,(i)}\right|^2 \frac{\left|\bxi_i^* {\bf R}_N {\bf Q}_{(i)}\bxi_i \tilde{r}_i^2\right|}{1+\frac{1}{n}\bxi_i^*\bQ \bxi_i}\right] \\
&\leq \frac{|z|}{n|\Im z|^4}\sum_{i=1}^n \mathbb{E}\left[\left|\stackrel{o}{\beta}_{i,(i)}\right|^2\bxi_i^*\bxi_i \right].
\end{align*}
Since $\beta_{i,(i)}$ is independent of $\bxi_i$, and thus :
\begin{align*}
\Delta_{1,2}& \leq \frac{|z|}{n|\Im z|^4} \sum_{i=1}^n \tr \bOmega_i\mathbb{E}|\stackrel{o}{\beta}_{i,(i)}|^2\\
&\leq \frac{N w_{\rm max}|z|}{n|\Im z|^4} \sum_{i=1}^n \mathbb{E}|\stackrel{o}{\beta}_{i,(i)}|^2
\end{align*}
From Lemma \ref{lemma:var_A}, we have:%we know that there exists a polynomial $P$ with positive coefficients such that:
$$
\mathbb{E}\left|\stackrel{o}{\beta}_{i,(i)}\right|^2 \leq \frac{2w_{\rm max}^3}{n^2} (|z|+1)\left(\frac{1}{|\Im z|^4}+\frac{1}{|\Im z|^3}\right)
$$
Hence, 
\begin{align*}
\Delta_{1,2} &\leq \limsup\frac{N}{n}\frac{2 w_{\rm max}^4}{n|\Im z|^4} (|z|+1)^2 \left(\frac{1}{|\Im z|^4}+\frac{1}{|\Im z|^3}\right)\\
&\triangleq \frac{K}{n} (|z|+1)^2 P(|\Im z|^{-1}),
\end{align*}
thereby proving the desired result.
The control of $\Delta_2$ relies on the use of the Cauchy-schwartz inequality. We have:
\begin{align*}
\Delta_2&=\frac{z}{n}\sum_{\ell=1}^n \mathbb{E}\left[\stackrel{o}{\beta}_\ell \tr \left({\bf Q}\bOmega_\ell{\bf R}_N\tilde{\bf R}_N^2\right)\right] \\
&\leq |z|\sum_{\ell=1}^n \sqrt{\mathbb{E}|\stackrel{o}\beta_\ell|^2}\sqrt{{\rm var}\frac{1}{n}\tr {\bf Q}\bOmega_\ell {\bf R}_N\tilde{\bf R}_N^2}
\end{align*}
From Lemma \ref{lemma:var_A}, we can bound $\mathbb{E}\left|\stackrel{o}{\beta}_\ell\right|^2$ and ${\rm var}\tr {\bf Q}\bOmega_\ell {\bf R}_N\tilde{\bf R}_N^2$ as:
\begin{align*}
\mathbb{E}\left|\stackrel{o}\beta_\ell\right|^2&\leq \frac{2w_{\rm max}^3}{n^2}(|z|+1)\left(\frac{1}{|\Im z|^4}+\frac{1}{|\Im z|^3}\right) \\
{\rm var}\tr \frac{1}{n}{\bf Q}\bOmega_\ell {\bf R}_N\tilde{\bf R}_N^2 &\leq \frac{2w_{\rm max}^3}{|\Im z|^6n^2}(|z|+1) \left(\frac{1}{|\Im z|^4}+\frac{1}{|\Im z|^3}\right).
\end{align*}
Using the fact that $\sqrt{xy}\leq \frac{x+y}{2}$ for positive scalars $x,y$, we finally get:
\begin{align*}
|\Delta_2| &\leq \frac{2w_{\rm max}^3(|z|+1)^2}{n} \left(\frac{1}{|\Im z|^4}+\frac{1}{|\Im z|^3}+\frac{1}{|\Im z|^{10}}+\frac{1}{|\Im z|^9}\right) \\
&\triangleq K_2(|z|+1)^2P_2(|\Im z|^{-1})
\end{align*}
Finally, we will move to the treatment of $\Delta_3$. Recall that $\Delta_3$ is given by:
$$
\Delta_3=\frac{z}{n}\sum_{\ell=1}^n \sum_{s=1}^N \tilde{r}_\ell^2\mathbb{E}\left[\left[\bSigma_n^*{\bf R}_N{\bf Q}\bOmega_\ell\right]_{\ell,s}\frac{\partial \stackrel{o}{\beta}_\ell}{\partial \xi_{s,\ell}^*}\right].
$$
Using the differentiation formulae in \eqref{eq:diff_1}, we get:
$$
\frac{\partial \stackrel{o}{\beta}_\ell}{\partial \xi_{s,\ell}^*} =-\frac{1}{n^2}\left[{\bf Q}\bOmega_\ell{\bf Q}\bSigma_n\right]_{s,\ell}.
$$
Hence,
\begin{align*}
\Delta_3&=-\frac{z}{n^3}\sum_{\ell=1}^n \sum_{s=1}^N \tilde{r}_\ell^2\mathbb{E}\left[\left[\bSigma_n^*{\bf R}_N{\bf Q}\bOmega_\ell\right]_{\ell,s} \left[{\bf Q}\bOmega_\ell{\bf Q}\bSigma_n\right]_{s,\ell}\right] \\
&=-\frac{z}{n^3}\sum_{\ell=1}^n \tilde{r}_\ell^2\mathbb{E}\left[\bxi_\ell^*{\bf R}_N{\bf Q}\bOmega_\ell {\bf Q}\bOmega_\ell{\bf Q}\bxi_\ell\right].
\end{align*}
The above relation allows us to bound $\Delta_3$ as:
\begin{align*}
|\Delta_3| &\leq \frac{|z|w_{\rm max}^2}{n^3} \sum_{\ell=1}^n |\tilde{r}_\ell|^2\|{\bf R}_N\|\mathbb{E}\left[\bxi_\ell^*\bxi_\ell \|{\bf Q}\|^3\right] \\
&\leq \frac{|z|w_{\rm max}^3}{n|\Im z|^6}\limsup\frac{N}{n} \\
&\triangleq \frac{K_3|z|}{n}P_3(|\Im z|^{-1}).
\end{align*}
From the obtained bounds for the scalars $\Delta_i,i=1,2,3$, we can deduce that:
$$
\left|z\Gamma\right| \leq \frac{1}{n}\left(|z|+C_1\right)^{k_1}P_1(|\Im z|^{-1}), 
$$
which is, as mentioned above, the required inequality to control $\chi_1$. 

%In order to control $\Delta_i,i=1,2,3$, we will need to bound the variance of quantities $\beta_{\bf A}(z)=\frac{1}{n}\tr {\bf A}{\bf Q}(z)$. 

\subsection{Control of $\chi_2(z)$}
We now move to the control of $\chi_2(z)$ given by:
$$
\chi_2(z)=N\tr {\bf R}_N-N\tr{\bf T}_N.
$$
To this end, we will resort to the resolvent identity : ${\bf A}^{-1}-{\bf B}^{-1}={\bf B}^{-1}\left({\bf A}-{\bf B}\right){\bf A}^{-1}$ for any invertible matrices ${\bf B}$ and ${\bf A}$. We therefore obtain:
\begin{align*}
N\tr{\bf R}_N-N\tr {\bf T}_N&=\frac{N}{n} \tr{\bf R}_N\left(\sum_{j=1}^n \frac{\bOmega_j}{1+\delta_j}-\frac{\bOmega_j}{1+\alpha_j}\right){\bf T}\\
&=\frac{N}{n}\sum_{j=1}^n \frac{\tr ({\bf R}_N\bOmega_j{\bf T})(\alpha_j-\delta_j)}{(1+\alpha_j)(1+\delta_j)}\\
&=\frac{N}{n}\sum_{j=1}^n z^2\tilde{r}_j\tilde{\delta}_j\tr {\bf R}_N\bOmega_j{\bf T}(\alpha_j-\delta_j),
\end{align*}
where $\tilde{\delta}_j=-\frac{1}{z(1+\delta_j)}$.
Using property 6 of Lemma 1 in \cite{LOU10}, we can easily check that $\tilde{\delta}_j,j=1,\cdots,n$ similar to $\tilde{r}_j$ are Stieltjes transforms of probability measures carried by $\mathbb{R}_{+}$. We therefore have:
$$
\max\left(\left|\tilde{\delta}_j\right|,\left|\tilde{r}_j\right|\right)\leq \frac{1}{|\Im z|}.
$$
Hence,
$$
\left|N \tr{\bf R}_N-N\tr{\bf T}_N\right|\leq \frac{|z|^2N^2}{|\Im z|^4} \max_{1\leq j\leq n}\left|\alpha_j-\delta_j\right|.
$$
To control $\chi_2$, it suffices to show that there exists constants $C$ and $K$, integer $k$ and polynomial $P$ with positive coefficients and independent of $N$ such that:
$$
\max_{1\leq j\leq n} \left|\alpha_j-\delta_j\right| \leq\frac{K}{N^2} (|z|+C)^k P(|\Im z|^{-1}).
$$
This will be the objective of the next derivations in this section.

We start by decomposing $\alpha_j-\delta_j$ as:
\begin{align*}
\alpha_j-\delta_j&=\frac{1}{n}\tr \bOmega_j\mathbb{E}{\bf Q}- \frac{1}{n}\tr \bOmega_j {\bf R} +\frac{1}{n}\tr \bOmega_j{\bf R} -\frac{1}{n}\tr \bOmega_j{\bf T} \\
&=\epsilon_j(z)+\frac{1}{n}\tr \bOmega_j{\bf R} -\frac{1}{n}\tr \bOmega_j{\bf T}.
\end{align*}
The control of $\epsilon_j(z)$ is similar to that of $\chi_1(z)$, the presence of matrix $\bOmega_j$ instead of the identity matrix requiring only slight modifications of the proof. We can thus deduce that:
\begin{equation}
\max_{1\leq j\leq n}\left|\epsilon_j\right|\leq \frac{K_\epsilon}{N^2}(|z|+C_\epsilon)^{k_\epsilon}P_\epsilon(|\Im z|^{-1}),
\label{eq:epsilon}
\end{equation}
for some constants $K_\epsilon$ and $C_\epsilon$, integer $k_\epsilon$ and polynomial $P_\epsilon$ independent of $N$.
Again, using the resolvent identity as above, we obtain:
\begin{equation}
\alpha_j-\delta_j=\epsilon_j(z)+\frac{1}{n^2}\sum_{k=1}^n \frac{(\alpha_k-\delta_k)\tr \bOmega_j{\bf R}_N\bOmega_k{\bf T}}{(1+\alpha_k)(1+\delta_k)}.
\label{eq:alpha-delta}
\end{equation}
Define $\boldsymbol{\alpha}=\left[\alpha_1,\cdots,\alpha_n\right]^{\mbox{\tiny T}}$, $\boldsymbol{\delta}=\left[\delta_1,\cdots,\delta_n\right]^{\mbox{\tiny T}}$ and $\boldsymbol{\epsilon}=\left[\epsilon_1(z),\cdots,\epsilon_n(z)\right]$. Then \eqref{eq:alpha-delta} writes as:
\begin{equation}
\left({\bf I}_n-{\bf A}\right)\left(\boldsymbol{\alpha}-\boldsymbol{\delta}\right)=\boldsymbol{\epsilon},
\label{eq:linear_system}
\end{equation}
where ${\bf A}$ is a $n\times n$ matrix with entries:
$$
\left[{\bf A}\right]_{j,k}=\frac{1}{n^2}\frac{\tr \bOmega_j{\bf R}_N\bOmega_k{\bf T}}{(1+\alpha_k)(1+\delta_k)}.
$$
In order to control the difference vector $\boldsymbol{\alpha}-\boldsymbol{\delta}$, we need first to check that ${\bf I}_n-{\bf A}$ is invertible. For that, notice that by Cauchy-Schwartz inequality:
$$
\left|\left[{\bf A}\right]_{j,k}\right| \leq \sqrt{\left|\left[{\bf B}\right]_{j,k}\right|} \sqrt{\left|\left[{\bf C}\right]_{j,k}\right|}
$$ 
where ${\bf B}$ and ${\bf C}$ are $n\times n$ matrices with entries:
\begin{align*}
\left[{\bf B}\right]_{j,k}&=\frac{1}{n^2}\frac{\tr\bOmega_j{\bf R}_N\bOmega_k{\bf R}_N}{\left|1+\alpha_k\right|^2} \\
\left[{\bf C}\right]_{j,k}&=\frac{1}{n^2}\frac{\tr\bOmega_j{\bf T}\bOmega_k{\bf T}}{\left|1+\delta_k\right|^2}.
\end{align*}
It follows from the algebraic lemma proven in Appendix \ref{app:algebraic} that ${\bf I}_n-{\bf A}$ is invertible provided that ${\bf B}$ or ${\bf C}$ have spectral norms strictly less than $1$, in which case:
\begin{equation}
\left\|\left({\bf I}_n-{\bf A}\right)^{-1}\right\|_{\infty}\leq \sqrt{\left\|\left({\bf I}_n-{\bf B}\right)^{-1}\right\|_{\infty}}\sqrt{\left\|\left({\bf I}_n-{\bf C}\right)^{-1}\right\|_{\infty}}.
\label{eq:relation}
\end{equation}
It appears from \eqref{eq:relation} that one needs to study matrices ${\bf B}$ and ${\bf C}$, which are at first sight easier to manipulate, mainly because they either involve ${\bf R}_N$ or ${\bf T}$. This however is not trivial. We state the result in the following proposition and for sake of readability defer the proof to Appendix \ref{app:BC}.
\begin{proposition}
Assume that $z\in\mathbb{C}_{+}$. Then, 
\begin{enumerate}
\item Matrix ${\bf C}$ satisfies $\rho({\bf C}) <1$. Moreover,
\begin{equation}
\left\|\left({\bf I}_n-{\bf C}\right)^{-1}\right\|_{\infty} \leq \frac{K(\eta^2+|z|^2)^2}{|\Im z|^4}
\label{eq:first_desired_1}
\end{equation}
where $K$ and $\eta$ are some positive constants independent of $N$. 
\item There exists $2$ polynomials $Q_1$ and $Q_2$ independent of $N$ with positive coefficients such that for $N$ large enough and $z\in \mathcal{E}_N$ given by
$$
\mathcal{E}_N=\left\{z\in\mathbb{C}_{+}, \frac{1}{N^2}Q_1(|z|)Q_2(|\Im z|^{-1}) \leq \frac{1}{2}\right\}
$$
we have $\rho({\bf B})\leq 1$ and:
$$
\|({\bf I}_n-{\bf B})^{-1}\| \leq \tilde{K}\frac{(\tilde{\eta}^2+|z|^2)^2}{|\Im z|^4}.
$$
\end{enumerate}
\label{prop:BC}
\end{proposition} 
It follows from proposition \ref{prop:BC} that the spectral norm of ${\bf A}$ is strictly less than $1$. Thus, ${\bf I}_n-{\bf A}$ is invertible and for $z\in\mathcal{E}_N$,
\begin{align}
\|\left({\bf I}_n-{\bf A}\right)^{-1}\|_{\infty}& \leq\frac{1}{2}\left\|\left({\bf I}_n-{\bf B}\right)^{-1}\right\|_{\infty}+\frac{1}{2}\left\|\left({\bf I}_n-{\bf C}\right)^{-1}\right\|_{\infty}\nonumber\\
&\leq \frac{K_{\rm max}(\eta_{\rm max}^2+|z|^2)}{\left|\Im z\right|^4}, \label{eq:inequality}
\end{align}
where $K_{\rm max}=\max(K,\tilde{K})$ and $\eta_{\rm max}=\max(\eta,\tilde{\eta})$. Plugging \eqref{eq:inequality} into \eqref{eq:linear_system}, we obtain:
$$
\left\|\boldsymbol{\alpha}-\boldsymbol{\delta}\right\|_{\infty}\leq \frac{K_{\rm max}K_\epsilon}{N^2}(|z|+C_\epsilon)^{k_\epsilon}(\eta_{\rm max}+|z|^2)\frac{P_\epsilon(|\Im z|^{-1})}{|\Im z|^4},
$$
where the right hand side of the above inequality  can be put under the form:
$$
\frac{\overline{K}(\overline{C}+|z|^2)^k}{N^2} \overline{P}(|\Im z|^{-1}).
$$
for $\overline{K}$ and $\overline{C}$ positive constants, $\overline{k}$ integer, and $\overline{P}$ some polynomial with positive coefficients.
Consider now the case where $z\in\mathbb{C}_{+}\backslash\mathcal{E}_N$. We first remark that:
$$
\left|\alpha_j-\delta_j\right| \leq |\alpha_j|+|\delta_j|\leq\frac{2w_{\rm max}}{|\Im z|}.
$$
Since $z\notin \mathcal{E}_N$, we therefore have:
$$
\frac{1}{N^2}Q_1(|z|)Q_2(|\Im z|^{-1}) \geq \frac{1}{2}.
$$ 
Hence:
$$
\left\|\boldsymbol{\alpha}-\boldsymbol{\delta}\right\|_{\infty} \leq \frac{4w_{\rm max}}{|\Im z|N^2}Q_1(|z|)Q_2(|\Im z|^{-1})
$$
As a consequence, we can find for $C,K$ constants,  $k$ integer and $P$ polynomial with positive coefficients such that:
$$
\left\|\boldsymbol{\alpha}-\boldsymbol{\delta}\right\|_{\infty} \leq  \frac{K}{N^2} (|z|+C)^k P(|\Im z|^{-1}),
$$
thereby ending the proof.
\appendices
\section{Preliminaries}
Many of the results of the appendix part are based on the following key lemmas, which we recall in this section for sake of clarity.
\label{sec:prelim}
\begin{lemma}
\label{lemma:ABx}
	Let ${\bf A}=(a_{\ell,m})_{\ell,m=1}^n$ be an $n\times n$ real matrix and ${\bf u}$ and ${\bf v}$ be two real $n\times 1$ vectors. Assume that the entries of ${\bf A}$ are positive and that of ${\bf u}$ and ${\bf v}$ strictly positive. Assume, furthermore, that the equation:
	$$
	{\bf u}={\bf A}{\bf u}+{\bf v}
	$$
	is satisfied. Then, the spectral radius $\rho({\bf A})$ of ${\bf A}$ of ${\bf A}$ satisfies:
$$
\rho({\bf A})\leq 1-\frac{\min(v_{\ell})}{\max(u_\ell)} <1.
$$
\end{lemma}
\begin{lemma}[Matrix Inequality]
\label{lemma:jensen}
	Let ${\bf A}$ be a $n\times n$ hermitian matrix. Then,
	$$
	\frac{1}{n}\tr {\bf A}{\bf A}^{*} \geq \left|\frac{1}{n} \tr {\bf A}\right|^2
	$$
	with equality only if ${\bf A}$ is proportional to identity.
\end{lemma}
\begin{proof}
	Let ${\bf A}={\bf U}\boldsymbol{\Lambda}{\bf U}^{\mbox{\tiny H}}$ be an eiengevalue decomposition of ${\bf A}$. Consider $\lambda_1,\cdots,\lambda_n$ the eigenvalues of ${\bf A}$. Then, if there is $i\neq j$ such that $\lambda_i \neq \lambda_j$, we have due to the strict-convexity of $x\mapsto x^2$:
	\begin{align*}
	\frac{1}{n}\tr {\bf A}{\bf A}^*&=\frac{1}{n}\sum_{i=1}^n \lambda_i^2 \\
									  &> \left|\frac{1}{n}\sum_{i=1}^n \lambda_i\right|^2
	\end{align*}
\end{proof}

\section{Proof of Theorem \ref{th:zero}}
\label{app:th_zero}
In order to establish that $0$ does not belong to the support $\mathcal{S}_N$, we show that it exists $\epsilon >0$ for which $\mu_N\left(\left[0,x\right]\right)=0$ for each $x\in \left]0,\epsilon\right[$. To this end, define function $\phi:\mathbb{R}_{+}^n\times \mathbb{R}^{+} \to \mathbb{R}_{+}^n$, with:
$$
\phi(x_1,\cdots,x_n,z)=\left(\phi_1(x_1,\cdots,x_n,z),\cdots,\phi_n(x_1,\cdots,x_n,z)\right)
$$
where $\phi_i:\mathbb{R}_{+}^n\times \mathbb{R}_{+}\to \mathbb{R}_+$  is given by:
$$
\phi_i(x_1,\cdots,x_n,z)=\frac{1}{n}\tr \bOmega_i\left(\frac{1}{n}\sum_{k=1}^n \frac{\bOmega_k}{1+x_k}-z I_N\right)^{-1}.
$$
We need to show that there exists  ${\ell}_1,\cdots,{\ell}_n$ such that:
$$
\phi_i({\ell}_1,\cdots,{\ell}_n,0)=\ell_i.
$$
%or equivalently:
%$$
%\overline{x}_i=\frac{1}{n}\tr \bOmega_i\left(\frac{1}{n}\sum_{k=1}^n \frac{\bOmega_k}{1+x_k}\right)^{-1}
%$$
%Note that without loss of generality, one can assume that:
%$$
%\max_{1\leq i\leq n}\left\|\Omega_i\right\|< 1
%$$
%since $\phi_i(x_1,\cdots,x_n,0)$ remains invariant if we multiply $\bOmega_i$ by the same constant.
Let $p\in\mathbb{N}$ and $r_p=-\frac{1}{p}$. We will first start by proving that for each $p$, there exists a unique $\overline{x}_1^{p},\cdots,\overline{x}_n^{p}$  such that: 
$$
\phi_i(\overline{x}_1^{p},\cdots,\overline{x}_n^{p},r_p)=\overline{x}_i^{p}.
$$
For that, it suffices   to show that $\tilde{\phi}_p:\mathbb{R}_{+}^n\to \mathbb{R}_{+}^n, (x_1,\cdots,x_n)\mapsto \phi(x_1,\cdots,x_n,r_p)$ is a standard interference function. In particular, we need to check that $\phi$ satisfy the following properties:
\begin{itemize}
\item Nonnegativity: For each $x_1,\cdots,x_n \geq 0$  and each $i$ and $p$, $\phi_i(x_1,\cdots,x_n,r_p) >0$.
\item  Monotonicity: For each $x_1\geq x_1^{'},\cdots,x_n\geq x_n^{'}$, and each $i$ and $p$,
$$
\phi_i\left(x_1,\cdots,x_n,r_p\right)\geq \phi_i\left(x_1^{'},\cdots,x_n^{'},r_p\right).
$$
\item Scalability: For each $\alpha >1$, and each $i$ and $p$, $\alpha \phi_i(x_1,\cdots,x_n,r_p) > \phi_i(\alpha x_1,\cdots,\alpha x_n,r_p)$.
\end{itemize}
The first item is obvious since $\bOmega_i$ are positive definite matrices, while the second one follows from the fact that for positive definite matrices, ${\bf A}\succeq {\bf B}$ implies ${\bf B}^{-1}\succeq {\bf A}^{-1}$. Finally, to prove the last item, note that for $\alpha >1$,
\begin{align*}
\phi_i(\alpha x_1,\cdots,\alpha x_n,r_p)&< \frac{1}{N}\tr \bOmega_i\left(\frac{1}{n}\sum_{k=1}^n \frac{\bOmega_k}{\alpha(1+x_k)}-\frac{r_p}{\alpha} I_N\right)^{-1} \\
&=\alpha \phi_i(x_1,\cdots,x_n,r_p).
\end{align*} 
Therefore, 
$$
\phi_i(\alpha x_1,\cdots,\alpha x_n,r_p) > \alpha \phi_i(x_1,\cdots,x_n,r_p).
$$
According to \cite[Theorem 2]{YAT95}, $\tilde{\phi}_p$ is a standard interference function. To prove that there exists a unique $\overline{x}_1^p,\cdots,\overline{x}_n^p$ satisfying:
$$
\overline{x}_i^p=\phi_i(\overline{x}_1^p,\cdots,\overline{x}_n^p),
$$
we need to check that there exits $x_1,\cdots,x_n$ such that:
$$
x_i > \phi_i(x_1,\cdots,x_n,r_p).
$$
This condition holds true, since $\phi_i(x_1,\cdots,x_n)\leq \frac{1}{r_p}$, and so increasing $x_i$ to infinity will satisfy the above inequality. 

Moreover, consider the sequence:
$$
x_i^{(t,p)}=\phi_i(x_1^{(t-1,p)},\cdots,x_n^{(t-1,p)}), \hspace{0.5cm} i=1,\cdots,n
$$ 
where $x_1^{(0,p)},\cdots,x_n^{(0,p)}$ are arbitrary positive reals. Then, ${\bf x}^{(t,p)}=\left(x_1^{(t,p)},\cdots,x_n^{(t,p)}\right)$ converge to $\overline{\bf x}^{p}=\left(\overline{x}_1^{p},\cdots,\overline{x}_n^{p}\right)$. 

From this, we can prove that for $p\geq q$, we have for each $i\in \left\{1,\cdots,n\right\}$,
$$
\overline{x}_i^{p}\geq \overline{x}_i^{q}.
$$
To this end, we will consider the sequence, 
$$
x_i^{(t,p)}=\phi_i(x_1^{(t-1,p)},\cdots,x_n^{(t-1,p)}), \hspace{0.5cm} i=1,\cdots,n
$$ 
where $x_i^{(0,p)}=\overline{x}_i^{q}$ and will show that for any $t$,
$$
x_i^{(t,p)}\geq \overline{x}_i^{q}.
$$
We will proceed by induction on $t$. For $t=0$, the result obviously holds. Assume that the resuld holds for any $k\leq t$, i.e, 
$$
x_i^{(k,p)}\geq \overline{x}_i^{q}, \hspace{0.5cm} i=1,\cdots,n \hspace{0.1cm} \textnormal{and} \hspace{0.1cm} k\leq t.
$$
And let us prove it for $t=k+1$. We have:
\begin{align*}
x_i^{(t+1,p)}&=\phi_i(x_1^{(t,p)},\cdots,x_n^{(t,p)},r_p) \\
&\geq \phi_i(x_1^{(t,p)},\cdots,x_n^{(t,p)},r_q)\\ 
&\stackrel{(a)}{\geq} \phi_i(\overline{x}_1^{q},\cdots,\overline{x}_n^{q},r_q) \\
&=\overline{x}_i^{q}.
\end{align*}
where $(a)$ follows since $\phi_i$ is increasing in each variable and $x_i^{(t,p)}\geq \overline{x}_i^{(q)}$ by the induction assumption.

We have therefore shown that for $p\geq q$, 
$$
\overline{x}_i^{p}\geq \overline{x}_i^{q}.
$$
As $p$ tends to infinity, $\overline{x}_i^{p}$ will converge to a limit $\ell_i \in \mathbb{R}_{+}\cup\left\{+\infty\right\}$. Assume that for $i \in \left\{1,\cdots,n\right\}$, $\ell_i \neq +\infty$. Then, one can easily see, that necessarily, $\ell_i \neq +\infty$ for any $i\in \left\{1,\cdots,n\right\}$. We will prove now, that the case of $\ell_i=+\infty$ for all $i=1,\cdots,n$ cannot hold. For this observe that:
\begin{align*}
\sum_{i=1}^n \frac{\overline{x}_i^{p}}{1+\overline{x}_i^{p}}&=\sum_{i=1}^n \frac{1}{n}\tr \frac{\bOmega_i}{1+\overline{x}_i^p}\left(\frac{1}{n}\sum_{k=1}^n \frac{\bOmega_k}{1+x_k^{p}}+r_p {\bf I}_N\right)^{-1}\leq N.
\end{align*}
Let $\overline{x}_{\rm min}^{p}=\min_{1\leq i \leq n} \overline{x}_i^{p}$. 
We have thus:
$$
\frac{\overline{x}_{\rm min}}{1+\overline{x}_{\rm min}}\leq \frac{N}{n}
$$
or equivalently:
$$
\overline{x}_{\rm min} \leq \frac{\frac{N}{n}}{1-\frac{N}{n}}.
$$
which is contradiction with the fact that $\ell_i=+\infty$ for all $i$. 

Recall now that:
$$
\phi_i(\overline{x}_1^{p},\cdots,\overline{x}_n^{p},r_p)=\overline{x}_i^{p}.
$$
Taking the limit in $p$, we thus get that:
$$
\phi_i(\ell_1,\cdots,\ell_n,0) =\ell_i, 
$$
or equivalently:
$$
\frac{1}{n}\tr \bOmega_i\left(\frac{1}{n}\sum_{k=1}^n \frac{\bOmega_k}{1+\ell_k}\right)^{-1}=\ell_i.
$$

The Jakobian matrix corresponding to $\tilde{\phi}_{\infty}:\mathbb{R}_{+}^n\to\mathbb{R}_+^{n}: (x_1,\cdots,x_n)\mapsto \phi(x_1,\cdots,x_n,0)$  at $x_i=\ell_i, i=1,\cdots,n$, is given by:
$$
\left[{\bf J}\right]_{i,m}=\frac{1}{n^2}\tr\bOmega_i \left(\frac{1}{n}\sum_{k=1}^n \frac{\bOmega_k}{1+\ell_k}\right)\frac{\bOmega_m}{(1+\ell_m)^2}\left(\frac{1}{n}\sum_{r=1}^n \frac{\bOmega_r}{1+\ell_r}\right)^{-1}
$$
Let ${\bf u}=\left[1+\ell_1,\cdots,1+\ell_n\right]^{\mbox{\tiny T}}$ and ${\bf v}=\left[\ell_1,\cdots,\ell_n\right]^{\mbox{\tiny T}}$.
 Then, after simple calculations, one can show that:
$$
{\bf J}{\bf u}={\bf v}.
$$
The entries of ${\bf J}$, ${\bf u}$ and ${\bf v}$ are strictly positive. A direct application of Lemma \ref{lemma:ABx} in section \ref{sec:prelim} implies that:
$$
\rho({\bf J})\leq 1-\frac{\min_{1\leq i\leq n} \ell_i}{1+\max_{1\leq i \leq n} \ell_i}<1.
$$
thereby showing that ${\bf I}_n-{\bf J}$ is invertible. Hence, the implicit function theorem implies that there exists an open disk at zero with radius $\eta>0$, i.e $D(0,\eta)$ and unique analytic functions $\varphi_1,\cdots,\varphi_n$ defined in $D(0,\eta)$ such that:
$$
\phi_i\left(\varphi_1(z),\cdots,\varphi_n(z),z\right)=\varphi_i(z)
$$
and 
$$
\varphi_i(0)=\ell_i,\hspace{0.2cm} i=1,\cdots,n.
$$
On the other hand, one can show that there exists $\epsilon >0$ such that $\varphi_i(t)$ is real valued and strictly positive for any $t\in \left[-\epsilon,\epsilon\right]$.  Indeed, writing $\Im\varphi_i(t)$ as:
\begin{align*}
\Im \varphi_i(t)&=\frac{1}{2\imath}\left(\frac{1}{n}\tr\bOmega_i\left(\frac{1}{n}\sum_{k=1}^n \frac{\bOmega_k}{1+\varphi_k(t)}-t{\bf I}_N\right)^{-1} \right.\\
&\left.- \frac{1}{n}\tr\bOmega_i\left(\frac{1}{n}\sum_{k=1}^n \frac{\bOmega_k}{1+\varphi_k^{*}(t)}-t {\bf I}_N\right)^{-1}\right)\\
&=\frac{1}{n}\tr \bOmega_i \left(\frac{1}{n}\sum_{k=1}^n \frac{\bOmega_k}{1+\varphi_k(t)}-t{\bf I}_N\right)^{-1}\\
 &\times  \left(\frac{1}{n}\sum_{k=1}^n \frac{\bOmega_k \Im (\varphi_k(t))}{\left|1+\varphi_k(t)\right|^2}\right)\left(\frac{1}{n}\sum_{k=1}^n \frac{\bOmega_k}{1+\varphi_k(t)}-t{\bf I}_N\right)^{-1}.
\end{align*}
Therefore, the vector ${\bf g}_t=\left[\Im(\varphi_1(t)),\cdots,\Im(\varphi_n(t))\right]^{\mbox{\tiny T}}$ is solution of the following system of equations:
$$
{\bf g}_t={\bf J}_t {\bf g}_t.
$$
As $t\mapsto\rho({\bf J}_t)$ is  continuous, and since for $t=0$, $\rho({\bf J}_t)=\rho({\bf J})<1$, there exists $\epsilon >0$ such that: 
$$
\rho({\bf J}_t) <1
$$
for every $t\in\left[-\epsilon,\epsilon\right]$. Therefore, ${\bf g}_t=0$. 
Furthermore, since at $t=0$, $\varphi_i(0)=\ell_i>0$, we can futher assume that $\epsilon$ is chosen such that $\varphi_i(t)$ is real-valued and strictly positive for any $t\in\left[-\epsilon,\epsilon\right]$. From \cite[Theorem 1]{WAG10}, we know that for $t<0$,  $\delta_1(t),\cdots,\delta_n(t)$ are the unique non-negative pointwise solutions of the following system of equations
$$
\delta_i(t)=\phi_i(\delta_1(t),\cdots,\delta_n(t),t),  
$$
thereby implying that:
$$
\delta_i(t)=\varphi_i(t) 
$$
for any $t\in\left[-\epsilon,0\right]$. Since, the set of functionals $\delta_1(t),\cdots,\delta_n(t)$ and $\varphi_1(t),\cdots,\varphi_n(t)$ are holomorphic on $D(0,\epsilon)\backslash\left\{\left[0,\epsilon\right[\right\}$ and coincide on a set of values with an accumulation point, they must coincide on the whole domaine of analicity, namely $D(0,\epsilon)\backslash\left\{\left[0,\epsilon\right[\right\}$.

Let $\overline{m}$ be given by:
$$
\overline{m}=\frac{1}{N}\tr \left(\frac{1}{n}\sum_{k=1}^n \frac{\bOmega_k}{1+\varphi_k(z)} - z{\bf I}_N\right)^{-1}.
$$
Obviously $\overline{m}$ is analytic on $D(0,\epsilon)$ and satisfies:
$$
\overline{m}(z)=m_N(z)
$$
for all $z\in D(0,\epsilon)\backslash\left\{\left[0,\epsilon\right[\right\}$.
We recall that for $0\leq x<\epsilon$, $\mu_N\left(\left[0,x\right]\right)$ can be expressed as:
$$
\mu_N\left(\left[0,x\right]\right)=\frac{1}{\pi}\lim_{y\to 0, y>0} \int_0^x \Im(m_N(s+\imath y)) ds.
$$
Therefore,
$$
\mu_N\left(\left[0,x\right]\right)=\frac{1}{\pi}\lim_{y\to 0, y>0} \int_0^x \Im(\overline{m}(s+\imath y)) ds
$$
As $\overline{m}$ is holomorphic on $D(0,\epsilon)$, the dominated convergence theorem implies that:
$$
\frac{1}{\pi}\lim_{y\to 0, y>0} \int_0^x \Im(\overline{m}(s+\imath y)) ds =\frac{1}{\pi}\int_0^x \Im (\overline{m}(s))ds =0
$$
since $\overline{m}(s)\in\mathbb{R}$ for $s\in \left[0,x\right]$. Thus, we establish that $\mu_N\left(\left[0,x\right]\right)=0$.

 \section{Proof of lemma \ref{lemma:var_A}}
\label{app:variance}
The proof follows from a direct application of the Nash-Poincar\'e inequality in Lemma \ref{lemma:nash}. Define $\beta_{\bf A}=\frac{1}{n}\tr{\bf A}{\bf Q}(z)$. We then have:
\begin{small}
\begin{align*}
{\rm var}(\beta_{\bf A}(z)) &\leq \sum_{k=1}^n \sum_{s=1}^N\sum_{r=1}^N \frac{1}{n^4} \mathbb{E}\left[\left[\bSigma_n^*{\bf Q}{\bf A}{\bf Q}\right]_{k,s} \left[\bOmega_k\right]_{s,r} \left[{\bf Q}^{*}{\bf A}^{*}{\bf Q}^{*}\bSigma_n\right]_{r,k}\right] \\
&+\sum_{k=1}^n \sum_{s=1}^N\sum_{r=1}^N \frac{1}{n^4} \mathbb{E}\left[\left[{\bf Q}^{*}{\bf A}^{*}{\bf Q}^{*}\bSigma_n\right]_{k,s} \left[\bOmega_k\right]_{s,r}\left[\bSigma_n^*{\bf Q}{\bf A}{\bf Q}\right]_{r,k}\right]\\
&=\sum_{k=1}^n \frac{1}{n^4}\mathbb{E}\left[\left[\bSigma_n^* {\bf Q}{\bf A}{\bf Q}\bOmega_k{\bf Q}^*{\bf A}^* {\bf Q}^{*}\bSigma_n\right]_{k,k}\right]\\
&+\sum_{k=1}^n \frac{1}{n^4}\mathbb{E}\left[\left[{\bf Q}^*{\bf A}^{*}{\bf Q}^* \bSigma_n \bOmega_k\bSigma_n^* {\bf Q}{\bf A}{\bf Q}\right]_{k,k}\right].
\end{align*}
\end{small}
Since $\bOmega_k \preceq w_{\rm max}{\bf I}_N$ with $w_{\rm max}=\sup_N\max_{1\leq k\leq n}\left\|\bOmega_k\right\|$,
we have:
\begin{align*}
	{\rm var}(\beta_{\bf A})(z) &\leq \frac{w_{\rm max}}{n^3}\tr \left({\bf Q}{\bf A}{\bf Q}{\bf Q}^*{\bf A}^*{\bf Q}^*\frac{\bSigma_n\bSigma_n^*}{n}\right) \\
												  &+\frac{w_{\rm max}}{n^3}\tr \left({\bf Q}^*{\bf A}^*{\bf Q}^*\frac{\bSigma_n\bSigma_n^*}{n}{\bf Q}{\bf A}{\bf Q}\right).
\end{align*}
Using the resolvent identity:
$$
{\bf Q}(z)\frac{\bSigma_n\bSigma_n^*}{n} = \frac{\bSigma_n\bSigma_n^*}{n} {\bf Q}(z)={\bf I}_N+z{\bf Q}(z),
$$ 
and the inequality $\|{\bf Q}(z)\|\leq \frac{1}{\left|\Im(z)\right|}$, we obtain:
\begin{align*}
	{\rm var}(\beta_{\bf A}(z))& \leq \frac{2w_{\rm max}\|{\bf A}\|^2}{n^2} \left(\frac{1}{\left|\Im(z)\right|^3} +\frac{\left|z\right|}{\left|\Im(z)\right|^4}\right) \\
												 &\leq \frac{2w_{\rm max}\|{\bf A}\|^2}{n^2} \left(\left|z\right|+1\right)\left(\frac{1}{\left|\Im(z)\right|^4}+\frac{1}{\left|\Im(z)\right|^3}\right).
\end{align*}

\section{Proof of proposition \ref{prop:BC}}
\label{app:BC}
In order to prove proposition \ref{prop:BC}, we need first to show that the sequence of measures $\mu_N$ is tight. To this end, we will follow the same steps as in \cite[Lemma C1]{HAC07}. Observe that:
\begin{align}
\int_0^{+\infty} \lambda\mu_N(d\lambda)&= \lim_{y\to+\infty} \Re\left[-\imath y\left(\imath y m_N(\imath y)+1\right)\right]\nonumber\\
&=\lim_{y\to+\infty} \Re \left[-\imath y\left(\imath y \frac{1}{N}\tr {\bf T}_N(\imath y)+1\right)\right].\label{eq:previous}
\end{align}
On the other hand:
$$
{\bf T}_N(\imath y)\left(\frac{1}{n}\sum_{k=1}^n \frac{\bOmega_k}{1+\delta_k(\imath y)}-\imath y{\bf I}_N\right)={\bf I}_N.
$$
Therefore,
\begin{align}
1+\frac{1}{N}\tr \imath y {\bf T}_N(\imath y)&= \frac{1}{n}\sum_{k=1}^n \frac{1}{N}\frac{\tr \bOmega_k {\bf T}_N(\imath y)}{1+\delta_k(\imath y)}\nonumber\\
&=\frac{1}{n}\sum_{k=1}^n \frac{1}{c_N}\frac{\delta_k(\imath y)}{1+\delta_k(\imath y)}. \label{eq:above}
\end{align}
Plugging \eqref{eq:above} into \eqref{eq:previous}, we finally get:
$$
\int_0^{+\infty} \lambda \mu_N(d\lambda)=\lim_{y\to+\infty} \frac{1}{n}\frac{1}{c_N}\sum_{k=1}^n \frac{\Re\left[-\imath y \delta_k(\imath y)\right]}{\left|1+\delta_k(iy)\right|^2}.
$$
Since $\delta_k$ are Stieltjes transforms of finite positive measures, we have:
$$
\lim_{y\to+\infty}\left|\delta_k(\imath y)\right| =0
$$
Moreover, we have $\lim_{y\to+\infty} -\imath y \delta_k(\imath y)= \frac{1}{n}\tr \bOmega_k$, thereby establishing that:
$$
\sup_{N}\int_0^{+\infty} \lambda \mu_N(d\lambda) < +\infty.
$$
The tightness of the sequence $\mu_N$ follows directly from the above inequality. In the same way, we can also show that the sequence of measures corresponding to the Stieltjes transforms $\frac{1}{N}\tr {\bf R}$ is also tight. These two results will be of fundamental importance in the proof of proposition \ref{prop:BC}. 

We now return to the proof of proposition\ref{prop:BC}:
	\begin{paragraph}{Proof of proposition \ref{prop:BC}-1)}
The proof is based on the use of Lemma \ref{lemma:ABx} in section \ref{sec:prelim}. For that, we need to find a linear system involving matrix ${\bf C}$. 
For $z\in \mathbb{C}_{+}$, we have:
	\begin{align*}
		\Im(\delta_j)&=\frac{1}{2\imath n}\left(\tr \bOmega_j{\bf T}-\tr \bOmega_j{\bf T}^{\mbox{\tiny H}}\right)\\
												   &=\frac{1}{n}\tr \bOmega_j{\bf T}\left(\frac{1}{n}\sum_{k=1}^n \frac{\bOmega_k \Im(\delta_k)}{\left|1+\delta_k\right|^2} +\Im(z){\bf I}_N\right){\bf T}^{\mbox{\tiny H}}\\
												   &=\frac{1}{n^2}\sum_{k=1}^n \frac{\tr \bOmega_j {\bf T}\bOmega_k {\bf T}^{\mbox{\tiny H}}}{\left|1+\delta_k\right|^2} \Im(\delta_k)+\Im(z)\frac{1}{n}\tr \bOmega_j {\bf T}{\bf T}^{\mbox{\tiny H}}.
	\end{align*}
	Let ${\bf I}_\delta$ and ${\bf c}$ be the $n\times 1$ vectors given by:
	\begin{align*}
		{\bf I}_{\delta}&=\left[\Im(\delta_1),\cdots,\Im(\delta_n)\right]^{\mbox{\tiny T}}\\
		{\bf c}&=\left[\frac{1}{n}\tr \bOmega_1{\bf T}{\bf T}^{\mbox{\tiny H}},\cdots,\frac{1}{n}\tr \bOmega_n{\bf T}{\bf T}^{\mbox{\tiny H}}\right]^{\mbox{\tiny T}},
\end{align*}
	Then:
	$$
	{\bf I}_{\delta}={\bf C} \hspace{0.04cm}{\bf I}_{\delta}+\Im(z){\bf c}.
	$$
	Since $\Im(\delta_j)> 0$ for all $j$ and $\Im z>0$ and ${\bf C}$, ${\bf c}$ have positive entries,  we get from Lemma \ref{lemma:ABx},
	\begin{align*}
		\left\|\left({\bf I}_n-{\bf C}\right)^{-1}\right\|_{\infty} &\leq \frac{\max_{1\leq j\leq n}\Im\delta_j}{\Im z \min_{1\leq j\leq n}\frac{1}{n}\tr \bOmega_j {\bf T}{\bf T}^{\mbox{\tiny H}}} \\
																	&\leq \frac{w_{\rm max}}{\left|\Im z\right|^2 \min_{1\leq j\leq n}\frac{1}{n}\tr \bOmega_j {\bf T}{\bf T}^{\mbox{\tiny H}}},
\end{align*}
where the second inequality follows from the fact that $\max_{1\leq j\leq n}\Im\delta_j\leq \max_{1\leq j\leq n}\left|\delta_j\right|\leq\frac{w_{\rm max}}{\Im z}$.
Using the inequality $\frac{1}{n}\tr {\bf A}{\bf B}\geq \lambda_1({\bf A})\frac{1}{n}\tr {\bf B}$ for ${\bf A}$ and ${\bf B}$ hermitian positive definite matrices with $\lambda_1({\bf A})$ the smallest eigenvalue of ${\bf A}$, we get:
\begin{equation}
\left\|\left({\bf I}_n-{\bf C}\right)^{-1}\right\|_{\infty}\leq \frac{w_{\rm max}}{\left|\Im z\right|^2w_{\rm min}\frac{1}{n}\tr {\bf T}{\bf T}^{\mbox{\tiny H}}}.
\label{eq:C}
\end{equation}
%	$$
%	\min_{1\leq j\leq n}\delta_j \geq u_j\min_{1\leq j \leq n}\delta_j +\Im (z)\frac{1}{n}\tr \bOmega_j {\bf T}{\bf T}^*.
%	$$
%	or equivalently,
%	$$
%	(1-u_j) \geq \frac{\Im(z)}{\displaystyle{\min_{1\leq j\leq n}\delta_j}} \frac{1}{n}\tr \bOmega_j{\bf T}{\bf T}^* \geq 0.
%	$$
%	Since $\displaystyle{\min_{1\leq j\leq n}\delta_j}\leq \frac{w_{\rm max}}{\Im (z)}$, we have:
%	\begin{equation}
%	(1-u_j)\geq \frac{\left|\Im z\right|^2w_{\rm min}}{w_{\rm max}} \frac{1}{n}\tr {\bf T}{\bf T}^{*}=\frac{\left|\Im z\right|^2 w_{\rm min}N}{w_{\rm max}n}\frac{1}{N}\tr {\bf T}{\bf T}^*.
%	\label{eq:final_previous}
%\end{equation}
	In order to obtain a lower bound on $\frac{1}{N}\tr {\bf T}{\bf T}^{\mbox{\tiny H}}$, we first remark that by the Jensen inequality in Lemma \ref{lemma:jensen}: $\frac{1}{N}\tr {\bf T}{\bf T}^{\mbox{\tiny H}} \geq \left|\frac{1}{N}\tr {\bf T}\right|^2= \left|m_N(z)\right|^2 \geq \Im (m_N(z))^2$. As $\left(\mu_N\right)_{N\geq 0}$ is tight, it exists $\eta >0$ for which $\mu_N(\left[\eta,+\infty\right)) \leq \frac{1}{2}$ for all $N$ and as such:
	$$
	\mu_N\left(\left[0,\eta\right]\right)\geq \frac{1}{2}.
	$$
	As a consequence,
	\begin{align}
		\Im(m_N(z))&=\Im(z)\int_{0}^{+\infty}\frac{d\mu_N(\lambda)}{\left|\lambda-z\right|^2} > \int_{0}^{\eta} \frac{\Im(z)d\mu_N(\lambda)}{2(\eta^2+|z|^2)}\mu_N(\left[0,\eta\right]) \nonumber\\
								  &\geq \frac{\Im(z)}{4(\eta^2+|z|^2)}.\label{eq:final}
\end{align}
Plugging \eqref{eq:final} into \eqref{eq:C}, we finally get \eqref{eq:first_desired_1}.
\end{paragraph}
\begin{paragraph}{Proof of proposition \ref{prop:BC}-2)}
	The proof is similar to that of the first statement. We first decompose $\alpha_j$ as:
$$
\alpha_j=\alpha_j-\frac{1}{n}\tr \bOmega_j{\bf R} +\frac{1}{n}\tr \bOmega_j{\bf R} =\epsilon_j +\frac{1}{n}\tr \bOmega_j{\bf R}. 
$$
Hence, 
	$$
	\Im (\alpha_j)=\Im(\epsilon_j(z)) +\Im \left(\frac{1}{n}\tr \bOmega_j{\bf R}\right).
	$$
	Using the same kind of calculations as above, we thus get:
	\begin{align}
		\Im (\alpha_j)&=\Im(\epsilon_j)+\frac{1}{n^2}\sum_{k=1}^n \frac{\tr\bOmega_j{\bf R}\bOmega_k{\bf R}^{\mbox{\tiny H}}\Im \alpha_k}{\left|1+\alpha_k(z)\right|^2} +\Im(z)\frac{1}{n}\tr \bOmega_j{\bf R}{\bf R}^{\mbox{\tiny H}}.\label{eq:alpha}
\end{align}
%or equivalently,
%$$
%(1-\tilde{u}_j)\min_{1\leq j\leq n} \Im(\alpha_j) \geq \Im \epsilon_j(z) +\Im(z)\frac{1}{n}\tr \bOmega_j{\bf R}{\bf R}^*.
%$$
In order to determine a subset of $\mathbb{C}_{+}$ on which $\Im(z)\frac{1}{n}\tr \bOmega_j{\bf R}{\bf R}^{\mbox{\tiny H}} +\Im(\epsilon_j(z)) > 0$, we evaluate a lower bound of $\frac{1}{n}\tr \bOmega_j{\bf R}{\bf R}^{\mbox{\tiny H}}$. We have by the Jensen inequality in Lemma \ref{lemma:jensen}:
$$
\frac{1}{n}\tr \bOmega_j {\bf R}{\bf R}^* \geq w_{\rm min}\left|\frac{1}{n}\tr {\bf R}\right|^2=w_{\rm min}\left(\frac{N}{n}\right)^2\left|\frac{1}{N}\tr {\bf R}\right|^2.
$$
From the discussion in the beginning of this section,  we know that the sequence of measures corresponding to the Stieltjes transforms $\frac{1}{N}\tr {\bf R}$ is tight. Hence, there exists $\tilde{\eta}$ such that:
$$
\Im \left(\frac{1}{N} \tr {\bf R}\right) \geq \frac{\Im z}{4\left(\tilde{\eta}^2+\left|z\right|^2\right)}.
$$
Hence,
$$
\frac{1}{n}\tr \bOmega_j{\bf R}{\bf R}^* \geq w_{\rm min} \left(\frac{N}{n}\right)^2\frac{\left|\Im z\right|^2}{16\left(\tilde{\eta}^2+\left|z\right|^2\right)^2}.
$$
On the other hand, from \eqref{eq:epsilon}, we recall that:
%By Lemma \ref{}, we know that there exists $P_1$ and $P_2$ polynomials with positive coefficients such that:
$$
\left|\epsilon_j(z)\right| \leq \frac{K_\epsilon}{N^2}(|z|+C_\epsilon)^{k_\epsilon}P_\epsilon(|\Im z|^{-1}).
%\frac{1}{N^2}P_1(\left|z\right|)P_2\left(\left|\Im z\right|^{-1}\right).
$$
Consider $\mathcal{E}_{N,1}$ the set given by:
\begin{align*}
\mathcal{E}_{N,1}&=\left\{z\in \mathbb{C}_{+}, \frac{w_{\rm min }\left(\frac{N}{n}\right)^2\left|\Im z\right|^2}{16(\tilde{\eta}^2+\left|z\right|^2)^2}\right.\\
&\left. -\frac{K_\epsilon}{N^2}(|z|+C_\epsilon)^{k_\epsilon}P_\epsilon(|\Im z|^{-1}) > 0\right\}
\end{align*}
Then, as before, using the fact that for $z\in\mathcal{E}_{N,1}$ \eqref{eq:alpha} can be cast into a  linear system of equations involving positive-entries matrix and vectors, we deduce that $\rho({\bf B}) <1$ and:
\begin{align*}
	&	\left\|\left({\bf I}_n-{\bf B}\right)^{-1}\right\|_{\infty} \leq \frac{\max_{1\leq j\leq n}\alpha_j}{\frac{w_{\rm min}N^2}{n^2}\frac{\left|\Im z\right|^3}{16(\tilde{\eta}^2+|z|^2)^2}-\frac{K_\epsilon}{N^2}(|z|+C_\epsilon)^{k_\epsilon}P_\epsilon(|\Im z|^{-1})}\\
																																																										   &\leq \frac{1}{\frac{w_{\rm min}N^2}{n^2}\frac{\left|\Im z\right|^4}{16\left(\tilde{\eta}^2+|z|^2\right)^2}\left(1-\frac{1}{N^2}Q_1(|z|)Q_2(\left|\Im z\right|^{-1})\right)},
\end{align*}
where $Q_1$ and $Q_2$ are polynomials with positive coefficients. 

	Take $\mathcal{E}_N$ as the set defined by:
			$$
			\mathcal{E}_N=\left\{z\in\mathbb{C}_{+},\frac{1}{N^2}Q_1\left(|z|\right)Q_2\left(\Im z^{-1}\right)\leq \frac{1}{2}\right\}.
			$$ 
			Obviously $\mathcal{E}_N\subseteq \mathcal{E}_{N,1}$, and for all $z\in\mathcal{E}_N$, we get:
			$$
		\left\|\left({\bf I}_n-{\bf B}\right)^{-1}\right\|_{\infty} \leq \frac{32n^2\left(\tilde{\eta}^2+|z|^2\right)^2}{{w_{\rm min}N^2}\left|\Im z\right|^4}.
			$$
		\end{paragraph}

\section{A linear algebraic result}
\label{app:algebraic}
Finally, we finish the Appendix part with a linear algebraic lemma which we need in our derivation and can be of independent interest.
\begin{lemma}
	Let ${\bf B}$ and ${\bf C}$ be $n\times n$ matrices with non-negative entries. %Assume that the spectral radius of ${\bf B}$ and ${\bf C}$ are such that $\max\left(\rho({\bf B}),\rho({\bf C})\right)\leq 1$. 
	Let ${\bf A}$ be a $n\times n$ matrix satisfying:
\begin{equation}
	\left|\left[{\bf A}\right]_{i,j}\right|\leq 	\sqrt{\left[{\bf B}\right]_{i,j}}	\sqrt{\left[{\bf C}\right]_{i,j}}.
\label{eq:strict}
	\end{equation}
	Then, $\rho({\bf A})\leq \sqrt{\rho({\bf B})}\sqrt{\rho({\bf C})}$. %If the inequality \eqref{eq:strict} is strict and $\max\left(\rho({\bf B}),\rho({\bf C})\right)\leq 1$, then $\rho({\bf A})<1$.
 If furthermore $\max(\rho({\bf A}),\rho({\bf B}))<1$, then  $\rho({\bf A})<1$ and:
	$$
	\left\|	\left({\bf I}_n-{\bf A}\right)^{-1}\right\|_{\infty} \leq \sqrt{	\left\|	\left({\bf I}_n-{\bf B}\right)^{-1}\right\|_{\infty}}\sqrt{	\left\|	\left({\bf I}_n-{\bf C}\right)^{-1}\right\|_{\infty}}
	$$
\end{lemma}
\begin{proof}
	We start by proving that $\rho({\bf A})\leq \sqrt{\rho({\bf B})}\sqrt{\rho({\bf C})}$. For that, consider $\tilde{\bf A}$, the matrix given by:
$$
\left[\tilde{\bf A}\right]_{i,j} = \sqrt{\left[{\bf B}\right]_{i,j}} \sqrt{\left[{\bf C}\right]_{i,j}}
$$
 %and recall that for any matrix ${\bf D}$,
%	$$
%	\rho({\bf D}) =\lim_{k\to+\infty} \|{\bf D}^k\|_{\infty}^{\frac{1}{k}}
%	$$
	Consider $\left|{\bf A}\right|$  the matrix such that $\left[\left|{\bf A}\right|\right]_{i,j}=\left|\left[{\bf A}\right]_{i,j}\right|$. Then,
$
\rho(\left|{\bf A}\right|) \leq \rho(\tilde{\bf A}).
$
%and $\rho(\left|{\bf A}\right|) <  \rho(\tilde{\bf A})$ if inequality \eqref{eq:strict} is strict.
 Recall, that for any matrix ${\bf D}$,
$$
	\rho({\bf D}) =\lim_{k\to+\infty} \|{\bf D}^k\|_{\infty}^{\frac{1}{k}}.
	$$
From the above convergence, we have:
	\begin{align*}
		&	\left[\tilde{\bf A}^k\right]_{i,j}=\sum_{i_1,\cdots,i_{k-1}}^n \left[\tilde{\bf A}\right]_{i,i_1}\left[\tilde{\bf A}\right]_{i_2,i_3}\cdots\left[\tilde{\bf A}\right]_{i_{k-1},j}\\
																												 &= \sum_{1\leq i_1,\cdots,i_{k-1}\leq n} \sqrt{\left[{\bf B}\right]_{i,i_1}\left[{\bf B}\right]_{i_2,i_3}\cdots\left[{\bf B}\right]_{i_{k-1},j}}\\
								   &\times\sqrt{\left[{\bf C}\right]_{i,i_1}\left[{\bf C}\right]_{i_2,i_3}\cdots\left[{\bf C}\right]_{i_{k-1},j}}\\
																													&\leq \sqrt{ \sum_{1\leq i_1,\cdots,i_{k-1}\leq n}\left[{\bf B}\right]_{i,i_1}\left[{\bf B}\right]_{i_2,i_3}\cdots\left[{\bf B}\right]_{i_{k-1,j}}}\\
								   &\sqrt{\sum_{1\leq i_1,\cdots,i_{k-1}\leq n}\left[ {\bf C}\right]_{i,i_1}\left[{\bf C}\right]_{i_2,i_3}\cdots\left[{\bf C}\right]_{i_{k-1},j}}\\
				 &=\sqrt{\left[{\bf B}^k\right]_{i,j}}
		\sqrt{\left[{\bf C}^k\right]_{i,j}}.
\end{align*}
With this inequality at hand, we are now in position to bound $\|\tilde{\bf A}^k\|_{\infty}$. 
We have:
\begin{align*}
	&	\left\|\tilde{\bf A}^k\right\|_{\infty}=\max_{1\leq i \leq n}\sum_{j=1}^n \left[\tilde{\bf A}^k\right]_{i,j} \\
																																			&\leq \max_{1\leq i \leq n}\sum_{j=1}^n \left[{\bf B}^k\right]_{i,j} \left[{\bf C}^k\right]_{i,j} \\
																							   &\leq  \max_{1\leq i \leq n}\sqrt{\sum_{j=1}^n \left[{\bf B}^k\right]_{i,j}} \sqrt{\sum_{j=1}^n \left[{\bf C}^k\right]_{i,j}} \\
																							   &\leq \sqrt{\left\|{\bf B}^k\right\|_{\infty}}\sqrt{\left\|{\bf C}^k\right\|_{\infty}}.
\end{align*}
We therefore have:
\begin{align*}
	\rho(\tilde{\bf A})&=\lim_{k\to+\infty} \left\|\tilde{\bf A}^k\right\|_{\infty}^{\frac{1}{k}}\\
							  &\leq \lim_{k\to+\infty}\left\|{\bf B}^k\right\|_{\infty}^{\frac{1}{2k}}\left\|{\bf C}^k\right\|_{\infty}^{\frac{1}{2k}} \\
						   &=\sqrt{\rho({\bf B})}\sqrt{\rho({\bf C})}.
\end{align*}
Therefore, $\rho(\tilde{\bf A})< 1$ and thus, $\rho({\bf A})<1$ if $\max(\rho({\bf C}),\rho({\bf B}))<1$. In this case, ${\bf I}_n-{\bf A}$ is  invertible and also are ${\bf I}_n-{\bf B}$ and ${\bf I}_n-{\bf C}$. Since $\left({\bf I}_n-{\bf A}\right)^{-1}=\sum_{k=0}^{+\infty} {\bf A}^k$. 
%Assume in addition that $\max(\rho({\bf B},{\bf C}))<1$, then ${\bf B}$ and ${\bf C}$ are invertible.
for any $1\leq i\leq n$, we have:
\begin{align*}
	&\sum_{j=1}^n\left|\left[\left({\bf I}_n-{\bf A}\right)^{-1}\right]_{i,j}\right|\leq \sum_{k=0}^{+\infty}\sum_{j=1}^n \left|\left[{\bf A}^k\right]_{i,j}\right|\\																																																																		&\leq \sum_{k=0}^{\infty}\sum_{j=1}^n\left[\left|{\bf A}\right|^k\right]_{i,j}\leq \sum_{k=0}^{+\infty}\sum_{j=1}^n \sqrt{\left[{\bf B}^k\right]_{i,j}}
 \sqrt{\left[\left|{\bf C}\right|^k\right]_{i,j}} \\
 &\leq \sum_{k=0}^{+\infty}\sqrt{\sum_{j=1}^n \left[{\bf B}^k\right]_{i,j} } \sqrt{ \sum_{j=1}^n\left[{\bf C}^k\right]_{i,j}} \\
 &\leq \sqrt{\sum_{k=0}^{+\infty}\sum_{j=1}^n \left[{\bf B}^k\right]_{i,j} } \sqrt{\sum_{k=0}^{+\infty}\sum_{j=1}^n \left[{\bf C}^k\right]_{i,j} }\\
 &\leq \sqrt{\left\|\left({\bf I}_n-{\bf B}\right)^{-1}\right\|_{\infty}}\sqrt{\left\|\left({\bf I}_n-{\bf C}\right)^{-1}\right\|_{\infty}}.
\end{align*}
As a consequence, we have:
$$
\left\|\left({\bf I}_n-{\bf A}\right)^{-1}\right\|_{\infty} \leq \sqrt{\left\|\left({\bf I}_n-{\bf B}\right)^{-1}\right\|_{\infty}}
\sqrt{\left\|\left({\bf I}_n-{\bf C}\right)^{-1}\right\|_{\infty}}.
$$
\end{proof}

%$$
%\max_{1\leq j \leq n} 1-\sqrt{u_j\tilde{u}_j}\left|\alpha_j-\delta_j\right|\leq \epsilon_j(z)
%$$
%From proposition \ref{}, for $z\in\mathcal{E}_N$, $1-u_j\geq 0$ and $1-\tilde{u}_j\geq 0$. Using the inequality $\sqrt{xy} \leq 1-\sqrt{(1-x)(1-y)}$ for any $0\leq x,y\leq 1$, we get:
%$$
%\sqrt{(1-u_j)(1-\tilde{u}_j)} \max_{1\leq j\leq n}\left|\alpha_j-\delta_j\right|\leq \epsilon_j(z).
%$$
%Hence, for any $z\in\mathcal{E}_N$, we have:
%$$
%\max_{1\leq j\leq n}\left|\alpha_j-\delta_j\right|\leq K\tilde{K}\epsilon_j(z)\frac{(\eta^2+|z|^2)(\tilde{\eta}^2+|z|^2)}{\left|\Im(z)\right|^4}
%$$
%Assume now that $z\in \mathbb{C}_{+}\backslash \mathcal{E}_N$. Then,
%$$
%2\leq \frac{2}{N^2}Q_1(\left|z\right|)Q_2(\left|\Im z\right|^{-1}).
%$$
%On the other hand, $\alpha_j$ $\delta_j$ are both less than $\frac{w_{\rm min}}{\Im(z)}$. Therefore,
%\begin{align*}
%	\max_{1\leq j\leq n}\left|\alpha_j-\delta_j\right|&\leq \frac{2w_{\rm min}}{\Im(z)}\\
%																																																   &\leq \frac{w_{\rm min}Q_1(|z|)Q_2(\left|\Im z\right|^{-1})}{\Im z}
%\end{align*}
%Evaluating the successive derivatives of functions $z\mapsto \phi_i(\varphi_1(z),\cdots,\varphi_n(z))$
%where ${\bf 1}_N$ is the vector of all ones. The entries of ${\bf }
\bibliographystyle{IEEEbib}
\bibliography{IEEEabrv,IEEEconf,tutorial_RMT}

\begin{thebibliography}{10}

\bibitem{LOU10}
P.~Vallet, P.~Loubaton, and X.~Mestre,
\newblock ``{Improved subspace estimation for multivariate observations of high
  dimension: the deterministic signals case},''
\newblock {\em {IEEE} Trans. Inf. Theory}, vol. 58, no. 2, Feb. 2012.

\bibitem{marchenko}
V.~A. Marchenko and L.~A. Pastur,
\newblock ``{Distributions of eigenvalues for some sets of random matrices},''
\newblock {\em Math. USSR Sb.}, vol. 1, no. 4, 1967.

\bibitem{SIL85}
J.~W. Silverstein,
\newblock ``{The smallest eigenvalue of a large dimensional wishart matrix},''
\newblock {\em The Annals of Probability}, vol. 13, no. 4, pp. 1364--1368, Nov.
  1985.

\bibitem{GEM80}
S.~Geman,
\newblock ``{A limit theorem for the norm of random matrices},''
\newblock {\em The Annals of Probability}, vol. 8, no. 2, pp. 252--261, Apr.
  1980.

\bibitem{BAI93}
Z.~D. Bai and Y.~Q. Yin,
\newblock ``{Limit of the smallest eigenvalue of a large dimensional sample
  covariance matrix},''
\newblock {\em The Annals of Probability}, vol. 21, no. 3, pp. 1275--1294, July
  1993.

\bibitem{BAI99}
Z.~D. Bai and J.~W. Silverstein,
\newblock ``{Exact separation of eigenvalues of large dimensional sample
  covariance matrices},''
\newblock {\em The Annals of Probability}, vol. 27, no. 3, pp. 1536--1555, July
  1999.

\bibitem{WAG10}
S.~Wagner, R.~Couillet, M.~Debbah, and D.~T.~M. Slock,
\newblock ``{Large system analysis of linear precoding in MISO broadcast
  channels with limited feedback},''
\newblock {\em {IEEE} Trans. Inf. Theory}, vol. 58, no. 7, pp. 4509--4537, July
  2012.

\bibitem{caire-13}
{A. Adhikary and J. Nam and J. Y. Ahn and G. Caire },
\newblock ``{Joint spatial division and multiplexing-the large-scale array
  regime},''
\newblock {\em {IEEE} Trans. Inf. Theory}, vol. 59, no. 10, pp. 6441--6463,
  Oct. 2013.

\bibitem{Kammoun2014b}
A.~Kammoun, A.~M\"{u}ller, E.~Bj\"{o}rnson, and M.~Debbah,
\newblock ``{Linear precoding based on polynomial expansion: large-scale
  multi-cell MIMO systems},''
\newblock {\em {IEEE} J. Sel. Topics Signal Process.}, vol. 8, no. 5, pp.
  861--875, Oct. 2014,
\newblock arXiv:1310.1799.

\bibitem{alouini-14}
M.~S. Alouini and A.~Kammoun,
\newblock ``{The random matrix regime of Maronna's M-Estimator for observations
  corrupted by elliptical noises},''
\newblock {\em To be submitted to Journal of Multivariate Analysis}, 2014.

\bibitem{HOC02}
B.~Hochwald and S.~Vishwanath,
\newblock ``{Space-time multiple access: Linear growth in the sum rate},''
\newblock {\em Proc. {IEEE} Annual Allerton Conference on Communication,
  Control, and Computing (Allerton'02)}, 2002.

\bibitem{couillet-pascal-2013}
R.~Couillet, F.~Pascal, and J.~W. Silverstein,
\newblock ``{The random matrix regime of Maronna's M-estimator with
  elliptically distributed samples},''
\newblock {\em submitted}, 2013,
\newblock \url{http://arxiv.org/abs/1311.7034}.

\bibitem{HAC07}
W.~Hachem, Ph. Loubaton, and J.~Najim,
\newblock ``{Deterministic equivalents for certain functionals of large random
  matrices},''
\newblock {\em Annals of Applied Probability}, vol. 17, no. 3, pp. 875--930,
  June 2007.

\bibitem{HAC06}
W.~Hachem, O.~Khorunzhy, P.~Loubaton, J.~Najim, and L.~A. Pastur,
\newblock ``{A new approach for capacity analysis of large dimensional
  multi-antenna channels},''
\newblock {\em {IEEE} Trans. Inf. Theory}, vol. 54, no. 9, pp. 3987--4004,
  Sept. 2008.

\bibitem{SIL98}
Z.~D. Bai and J.~W. Silverstein,
\newblock ``{No eigenvalues outside the support of the limiting spectral
  distribution of large dimensional sample covariance matrices},''
\newblock {\em Annals of Probability}, vol. 26, no. 1, pp. 316--345, Jan. 1998.

\bibitem{Paul-09}
P.~Debashis and J.~W. Silverstein,
\newblock ``{No eigenvalues outside the support of limiting empirical spectral
  distribution of a separable covariance matrix},''
\newblock {\em {Journal Of Multivariate Analysis}}, vol. 100, no. 1, Jan. 2009.

\bibitem{SIL12}
Z.~D. Bai and J.~W. Silverstein,
\newblock ``{No eigenvalues outside the support of the limiting spectral
  distribution of information-plus-noise type matrices},''
\newblock {\em Random Matrices: Theory and Applications}, , no. 2, pp. 1150004,
  Jan. 2012.

\bibitem{Capitaine-10}
M.~Capitaine, C.~Donati-Martin, and D.~F\'eral,
\newblock ``{The largest eigenvalue of finite rank deformation of large wigner
  matrices},''
\newblock {\em Annals of Probability}, vol. 37, no. 1, Jan. 2009.

\bibitem{haagerup}
U.~Haagerup and S.~Thorbj{\o}rnsen,
\newblock ``{A new application of random matrices: $\rm{Ext}(\rm{C}_{\rm
  red}^*(F2))$ is not a group},''
\newblock {\em {Annals of Mathematics}}, vol. 162, no. 2, pp. 711--775, Sept.
  2005.

\bibitem{SIL95}
J.~W. Silverstein and Z.~D. Bai,
\newblock ``{On the empirical distribution of eigenvalues of a class of large
  dimensional random matrices},''
\newblock {\em Journal of Multivariate Analysis}, vol. 54, no. 2, pp. 175--192,
  Aug. 1995.

\bibitem{YAT95}
R.~D. Yates,
\newblock ``{A framework for uplink power control in cellular radio systems},''
\newblock {\em {IEEE} J. Sel. Areas Commun.}, vol. 13, no. 7, pp. 1341--1347,
  Sept. 1995.

\end{thebibliography}
\end{document}